\documentclass[final,nomarks]{dmtcs-episciences}
\usepackage{amssymb}
\usepackage{amsthm}
\usepackage{array}

\usepackage[round]{natbib}

\newcommand{\e}{\varepsilon}
\newcommand{\dfempty}{\e}
\newcommand{\NN}[1]{\{1,\ldots,#1\}}
\newcommand{\dfplus}[1]{#1^{\diamond}}
\newcommand{\DFM}[2]{#1^{\diamond}_{#2}}
\newcommand{\DFX}[1]{#1^{\diamond}_{\bullet}}

\newcommand{\dfimglabel}[3]{
    \put(#1,#2){
        \put(0,0){\line(1,0){10}}
        \put(0,0){\line(0,1){10}}
        \put(10,10){\line(-1,0){10}}
        \put(10,10){\line(0,-1){10}}
        \put(3,2){$#3$}
    }
}

\newcommand{\dfimgbegin}[2]{
    \put(#1,#2){
        \put(5,5){\circle{10}}
    }
}

\newcommand{\dfimgend}[2]{
    \put(#1,#2){
        \put(0,0){{\huge $\diamond$}}
    }
}

\newcommand{\dfimgbeginlabel}[3]{
    \dfimglabel{#1}{#2}{#3}
    \dfimgbegin{#1}{#2}
}

\newcommand{\dfimgendlabel}[3]{
    \dfimglabel{#1}{#2}{#3}
    \dfimgend{#1}{#2}
}

\newcommand{\dfimgbeginendlabel}[3]{
    \dfimglabel{#1}{#2}{#3}
    \dfimgbegin{#1}{#2}
    \dfimgend{#1}{#2}
}

\newcommand{\dfbegin}[1]{\mbox{\rm begin}(#1)}
\newcommand{\dfend}[1]{\mbox{\rm end}(#1)}
\newcommand{\dftransition}[1]{\mbox{\rm tran}(#1)}
\newcommand{\dflabel}[1]{\mbox{\rm label}(#1)}
\newcommand{\dfdomain}[1]{\mbox{\rm dom}(#1)}

\newcommand{\CEP}[1]{\mbox{\rm CE}^+(#1)}
\newcommand{\CEN}[1]{\mbox{\rm CE}^-(#1)}
\newcommand{\CWP}[1]{\mbox{\rm CW}^+(#1)}
\newcommand{\CWN}[1]{\mbox{\rm CW}^-(#1)}
\newcommand{\B}[2]{\mbox{\rm B}(#1,#2)}
\newcommand{\HP}[2]{\mbox{\rm HP}(#1,#2)}
\newcommand{\R}[2]{\mbox{\rm Rot}_{#1}(#2)}
\newcommand{\tr}[1]{\mathrm{tr}_{#1}}
\newcommand{\hull}[1]{\mathrm{hull}(#1)}
\newcommand{\hullm}[1]{\mathrm{hull}^{*}(#1)}
\newcommand{\block}{{\mathrm B}_0}
\newcommand{\rblock}{{\mathrm B}_{\mathrm R}}
\newcommand{\lblock}{{\mathrm B}_{\mathrm L}}

\newcommand{\confBegends}[2]{\mathrm{EB}^{#1}_{\mathrm #2}}
\newcommand{\conf}[1]{\mathrm{C}_{\mathrm #1}}
\newcommand{\confDom}[2]{D^{#1}_{\mathrm #2}}
\newcommand{\confLab}[2]{l^{#1}_{\mathrm #2}}
\newcommand{\confEnd}{\odot}

\newcommand{\ddx}{\ddot{x}}
\newcommand{\ddy}{\ddot{y}}
\newcommand{\dx}{\dot{x}}
\newcommand{\dy}{\dot{y}}
\newcommand{\DHSP}[4]{
	\put(0,0){\line(1,0){50}}
	\put(0,0){\line(0,1){50}}
	\put(50,0){\line(0,1){50}}
	\put(0,50){\line(1,0){50}}
	
	\put(23,40){$#1$}
	\put(38,22){$#2$}
	\put(23,5){$#3$}
	\put(3,22){$#4$}
}
\newcommand{\DHSWE}{\put(15,25){\vector(1,0){20}}}
\newcommand{\DHSWS}{\put(15,25){\vector(1,-1){10}}}
\newcommand{\DHSEW}{\put(35,25){\vector(-1,0){20}}}
\newcommand{\DHSES}{\put(35,25){\vector(-1,-1){10}}}
\newcommand{\DHSNE}{\put(25,35){\vector(1,-1){10}}}
\newcommand{\DHSNW}{\put(25,35){\vector(-1,-1){10}}}

\newcommand{\DHSSE}{\put(25,15){\vector(1,1){10}}}
\newcommand{\DHSSN}{\put(25,15){\vector(0,1){20}}}
\newcommand{\DHSNS}{\put(25,35){\vector(0,-1){20}}}
\newcommand{\DHSWN}{\put(15,25){\vector(1,1){10}}}
\newcommand{\DHSSW}{\put(25,15){\vector(-1,1){10}}}
\newcommand{\DHSEN}{\put(35,25){\vector(-1,1){10}}}

\newcommand{\DHSh}[7]{\mbox{\rm DHS}_#1(#2,#3,#4,#5)^{#6}_{#7}}
\newcommand{\DHS}[6]{\mbox{\rm DHS}(#1,#2,#3,#4)^{#5}_{#6}}

\newcommand{\circTimes}[2]{\underbrace{#1\circ\cdots\circ #1}_{#2\mathrm{~times}}}

\newtheorem{definition}{Definition}
\newtheorem{theorem}{Theorem}
\newtheorem{example}{Example}

\newtheorem{lemma}{Lemma}
\newtheorem{proposition}{Proposition}
\newtheorem{corollary}{Corollary}

\author{W{\l}odzimierz Moczurad\thanks{Supported by National Science Centre (NCN) grant no.\ 2011/03/B/ST6/00418.}}

\title{Decidability of multiset, set and numerically decipherable directed figure codes}

\affiliation{Faculty of Mathematics and Computer Science, Jagiellonian University, 
Krak{\' o}w, Poland}


\received{2016-4-11}
\accepted{2017-4-12}
\revised{2017-1-11}
\publicationdetails{19}{2017}{1}{11}{1430}

\begin{document}

\maketitle

\begin{abstract}
Codes with various kinds of decipherability, weaker than the
usual unique decipherability, have been studied since multiset
decipherability was introduced in mid-1980s. We consider
decipherability of directed figure codes, where directed figures
are defined as labelled polyominoes with designated start and
end points, equipped with catenation operation that may use a
merging function to resolve possible conflicts. This is one of
possible extensions generalizing words and variable-length codes
to planar structures.
Here, verification whether a given set is a code is no longer
decidable in general. We study the decidability status of figure
codes depending on catenation type (with or without a merging
function), decipherability kind (unique, multiset, set or
numeric) and code geometry (several classes determined by
relative positions of start and end points of figures). We give
decidability or undecidability proofs in all but two cases that
remain open.
\end{abstract}

\section{Introduction}

The classical notion of a code requires that an encoded message
should be decoded uniquely, \textit{i.e.}\ the exact sequence of
codewords must be recovered. In some situations, however, it
might be sufficient to recover only the multiset, the set or
just the number of codewords. This leads to three kinds of
decipherability, known as \textit{multiset} (MSD), \textit{set}
(SD) and \textit{numeric decipherability} (ND), respectively.
The original exact decipherability is called \textit{unique
decipherability} (UD).

Multiset decipherability was introduced 
by~\cite{Lempel}, whilst numeric decipherability originates
in~\cite{HeadWeber1}. The same authors
in~\cite{HeadWeber2} develop what they call ``domino graphs'' providing a
useful technique for decipherability verification. 
\cite{Guzman} defined set decipherability and presented a
unifying approach to different decipherability notions using
varieties of monoids. Contributions by
\cite{Restivo} and \cite{BS1}
settle Lempel's conjectures for some MSD and SD codes.
\cite{BS2} characterizes decipherability of
three-word codes, whilst \cite{BurRest1,BurRest2} relate
decipherability to the Kraft inequality and to
coding partitions. A paper by~\cite{SSY}, although not directly concerned with
decipherability, uses ND codes (dubbed \textit{length codes}) to
study prime decompositions of languages.

Extensions of classical words and variable-length word codes
have also been widely studied. For instance, \cite{AigrBeauq}
introduced polyomino codes;
two-dimensional rectangular pictures were studied by 
\cite{GiamRest}, whilst \cite{MantRest}
described an algorithm to verify tree codes.
Recent results on picture codes include \textit{e.g.}\ \cite{AGM1,AGM2}.
The interest in picture-like structures is not surprising, given
the huge amounts of pictorial data in use. Unfortunately,
properties related to decipherability are often lost when moving
to a two-dimensional plane. In particular, decipherability
testing (\textit{i.e.}\ testing whether a given set is a code) is
undecidable for polyominoes and similar structures,
\textit{cf.}~\cite{Beauq,WMijcm}.

In \cite{KolMoc} we introduced directed figures defined as
labelled polyominoes with designated start and end points,
equipped with catenation operation that uses a merging function
to resolve possible conflicts. This setting is similar to
symbolic pixel pictures, described by \cite{CFG}, 
and admits a natural definition of
catenation. The attribute ``directed'' is used to emphasize the
way figures are catenated; this should not be confused with the
meaning of ``directed'' in \textit{e.g.}\ directed polyominoes.
We proved that verification whether a given finite set of
directed figures is a UD code is decidable. This still holds
true in a slightly more general setting of codes with weak
equality (see \cite{WMideal}) and is a significant change in
comparison to previously mentioned picture models, facilitating
the use of directed figures in, for instance, encoding and
indexing of pictures in databases. On the other hand, a directed
figure model with no merging function, where catenation of
figures is only possible when they do not overlap, has again
undecidable UD testing; \textit{cf.}~\cite{KolCOCOON,KolRAIRO}.
See also \cite{WMiwoca2013} for a short description of decipherability 
chracaterization with domino graphs.

In the present paper we extend the previous results by
considering not just UD codes, but also MSD, SD and ND codes
over directed figures. We prove decidability or undecidability
for each combination of the following orthogonal criteria:
catenation type (with or without a merging function),
decipherability kind (UD, MSD, SD, ND) and code geometry
(several classes determined by relative positions of start and
end points of figures). Two combinations remain open, however.

We begin, in Section~\ref{sec:defs}, with definitions of
directed figures and their catenations. Section~\ref{sec:codes}
defines decipherability kinds and shows the relationship between
codes of those kinds. In Section~\ref{sec:dec} main decidability
results for decipherability verification are given. Preliminary, short version
of this paper appeared as~\cite{KolMocIWOCA}.

\section{Preliminaries}\label{sec:defs}

Let $\Sigma$ be a finite, non-empty alphabet. A
\emph{translation} by vector $u=(u_x,u_y)\in\integers^2$ is
denoted by $\tr{u}$,
$\tr{u}:\integers^2\ni(x,y)\mapsto(x+u_x,y+u_y)\in\integers^2$.
By extension, for a set $V\subseteq\integers^2$ and an arbitrary
function $f:V\to\Sigma$ define
$\tr{u}:\mathcal{P}(\integers^2)\ni V \mapsto \{\tr{u}(v)\mid
v\in V\}\in \mathcal{P}(\integers^2)$ and $\tr{u}:\Sigma^V\ni f
\mapsto f\circ \tr{-u}\in \Sigma^{\tr{u}(V)}$.


\begin{definition}[Directed figure, \textit{cf.}~\cite{KolMoc}]
Let $D\subseteq\integers^2$ be finite and non-empty, $b,e\in\integers^2$
and $\ell:D\to\Sigma$. A~quadruple $f=(D,b,e,\ell)$ is
a \emph{directed figure} (over $\Sigma$) with
\begin{center}
\begin{tabular}{lrcl}
\emph{domain}      & $\dfdomain{f}$ &=& $D$,\\
\emph{start point} & $\dfbegin{f}$  &=& $b$,\\
\emph{end point}   & $\dfend{f}$     &=& $e$,\\
\emph{labelling function} & $\dflabel{f}$ &=& $\ell$.
\end{tabular}
\end{center}
\emph{Translation vector} of $f$ is defined as
$\dftransition{f}=\dfend{f}-\dfbegin{f}$. Additionally, the \emph{empty
directed figure} $\e$ is defined as $(\emptyset,(0,0),(0,0),\{\})$,
where $\{\}$ denotes a function with an empty domain.
Note that the start and end points need not be in the domain.
\end{definition}

The set of all directed figures over $\Sigma$ is denoted by
$\dfplus{\Sigma}$. Two directed figures $x,y$ are \emph{equal} (denoted
by $x=y$) if there exists $u\in\integers^2$ such that
\begin{displaymath}
   y=(\tr{u}(\dfdomain{x}),\tr{u}(\dfbegin{x}),\tr{u}(\dfend{x}),\tr{u}(\dflabel{x})).
\end{displaymath}
Thus, we actually consider figures up to translation.

\begin{example}
A directed figure and its graphical representation. Each point of
the domain, $(x,y)$, is represented by a unit square in $\reals^2$
with bottom left corner in $(x,y)$. A circle marks the start point
and a diamond marks the end point of the figure. Figures are
considered up to translation, hence we do not mark the coordinates.
\begin{displaymath}
    (\{(0,0),(1,0),(2,0),(1,1)\},(0,1),(2,1),\{(0,0)\mapsto a,(1,0)\mapsto b,(1,1)\mapsto a,(2,0)\mapsto a\})
\end{displaymath}
\begin{center}
\begin{picture}(30,20)
    \dfimgbegin{00}{10}
    \dfimglabel{00}{00}{a}
    \dfimglabel{10}{00}{b}
    \dfimglabel{10}{10}{a}
    \dfimglabel{20}{00}{a}
    \dfimgend{20}{10}
\end{picture}
\end{center}
\end{example}

\begin{definition}[Catenation, \textit{cf.}~\cite{KolMoc}]
Let $x=(D_x,b_x,e_x,\ell_x)$ and $y=(D_y,b_y,e_y,\ell_y)$ be directed figures.
If $D_x\cap\tr{e_x-b_y}(D_y)=\emptyset$,
a \emph{catenation} of $x$ and $y$ is defined as
\begin{displaymath}
x\circ y=(D_x\cup\tr{e_x-b_y}(D_y),b_x,\tr{e_x-b_y}(e_y),\ell),
\end{displaymath}
where
\begin{displaymath}
\ell(z)=\left\{\begin{array}{ll}
\ell_x(z) & \mbox{~for~} z\in D_x,\\
\tr{e_x-b_y}(\ell_y)(z) & \mbox{~for~} z\in \tr{e_x-b_y}(D_y).\\
\end{array}\right.
\end{displaymath}
If $D_x\cap\tr{e_x-b_y}(D_y) \neq \emptyset$, catenation of $x$ and $y$ is not defined.
\end{definition}

\begin{definition}[$m$-catenation, \textit{cf.}~\cite{KolMoc}]
Let $x=(D_x,b_x,e_x,\ell_x)$ and $y=(D_y,b_y,e_y,\ell_y)$ be directed figures.
An \emph{$m$-catenation} of $x$ and $y$ with respect to a \emph{merging function}
$m:\Sigma\times\Sigma\to\Sigma$ is defined as
\begin{displaymath}
x\circ_m y=(D_x\cup\tr{e_x-b_y}(D_y),b_x,\tr{e_x-b_y}(e_y),\ell),
\end{displaymath}
where
\begin{displaymath}
\ell(z)=\left\{\begin{array}{ll}
\ell_x(z)  & \mbox{~for~} z\in D_x\setminus\tr{e_x-b_y}(D_y),\\
\tr{e_x-b_y}(\ell_y)(z)  & \mbox{~for~} z\in \tr{e_x-b_y}(D_y)\setminus D_x,\\
m(\ell_x(z),\tr{e_x-b_y}(\ell_y)(z)) & \mbox{~for~} z\in D_x\cap\tr{e_x-b_y}(D_y).\\
\end{array}\right.
\end{displaymath}
\end{definition}

Notice that when $x\circ y$ is defined, it is equal to $x\circ_m y$, regardless
of the merging function~$m$.

\begin{example}
Let $\pi_1$ be the projection onto the first argument.
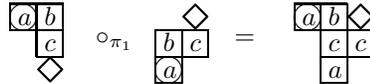
\begin{displaymath}
\begin{picture}(20,25)
    \dfimgbeginlabel{00}{20}{a}
    \dfimglabel{10}{20}{b}
    \dfimglabel{10}{10}{c}
    \dfimgend{10}{00}
\end{picture}~
\begin{picture}(30,25)
    \put(10,15){$\circ_{\pi_1}$}
\end{picture}~
\begin{picture}(20,25)
    \dfimgbeginlabel{00}{00}{a}
    \dfimglabel{00}{10}{b}
    \dfimglabel{10}{10}{c}
    \dfimgend{10}{20}
\end{picture}
\begin{picture}(30,25)
    \put(10,15){$=$}
\end{picture}~
\begin{picture}(30,25)
    \dfimgbeginlabel{00}{20}{a}
    \dfimglabel{10}{20}{b}
    \dfimglabel{10}{10}{c}
    \dfimglabel{10}{00}{a}
    \dfimglabel{20}{10}{c}
    \dfimgend{20}{20}
\end{picture}
\end{displaymath}
The ``non-merging'' catenation is not defined for the above figures. Note that the result 
of ($m$-)catenation does not depend on the original position of the second argument.
\end{example}

Observe that $\circ$ is associative, whilst $\circ_m$ is
associative if and only if $m$ is associative. Thus for
associative~$m$, $\DFM{\Sigma}{m}=(\dfplus{\Sigma},\circ_m)$ is a
monoid (which is never free).
From now on let $m$ be an arbitrary associative merging
function.

Abusing this notation, we also write $\dfplus{X}$ (resp.
$\DFM{X}{m}$) to denote the set of all figures that can be
composed by $\circ$ catenation (resp. $\circ_m$ $m$-catenation)
from figures in $X\subseteq\dfplus{\Sigma}$. When some statements
are formulated for both $\circ$ and $\circ_m$, we use the
symbol~$\bullet$ and ``$x\bullet y$'' should then be read as
``$x\circ y$ (resp.\ $x\circ_m y$)''. Similarly,
``$x\in\DFX{X}$'' should be read as ``$x\in\dfplus{X}$ (resp.\
$x\in\DFM{X}{m}$)''.

For $u,v\in\integers^2$, $\HP{u}{v}$ denotes a half-plane 
$\{w\in\integers^2\mid u\cdot (w-(v+u)) \leq 0\}$,
where $\cdot$ is the usual scalar product; see Figure~\ref{fig:HP}.
An angle between two vectors $u,v\in\integers^2$ is written as $\angle(u,v)$ and 
$\R{\phi}{u}$ denotes a rotation of~$u$ by an angle~$\phi$.
For $u=(u_x,u_y)\in\integers^2$ and $n\in\naturals$,
$\B{u}{n}$ denotes a ball on the integer grid with center~$u$ and
radius~$n$, \textit{i.e.},
$\B{u}{n}   = \{(v_x,v_y)\in\integers^2\mid\ |u_x-v_x|+|u_y-v_y| \leq n\}.$

\begin{figure}[htp]
 \begin{center}
 \begin{picture}(120,160)
 \thinlines

 \put(80,0){\line(0,1){160}}

 \put(80,10){\line(-1,0){120}}
 \put(80,30){\line(-1,0){120}}
 \put(80,50){\line(-1,0){120}}
 \put(80,70){\line(-1,0){120}}
 \put(80,90){\line(-1,0){120}}
 \put(80,110){\line(-1,0){120}}
 \put(80,130){\line(-1,0){120}}
 \put(80,150){\line(-1,0){120}}

 \put(10,77){$\bullet$}
 \put(78,77){$\bullet$}
 \put(12,80){\vector(1,0){67}}
 \put(5,72){$v$}
 \put(45,82){$u$}
 \put(82,72){$v+u$}

 \end{picture}
 \caption{$\HP{u}{v}$. The half-plane contains integer grid points lying on a vertical line 
 and to the left side of that line (the region marked by horizontal lines).}
 \label{fig:HP}
 \end{center}
 \end{figure}
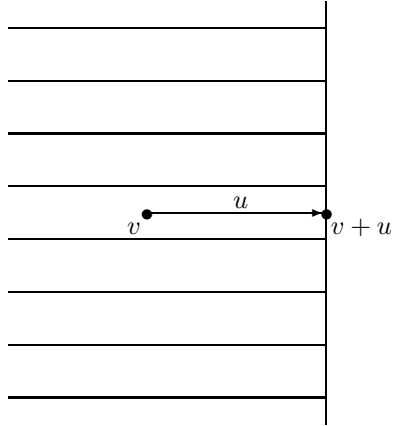

\section{Codes}
\label{sec:codes}

In this section we define a total of eight kinds of directed
figure codes, resulting from the use of four different notions
of decipherability and two types of catenation. Note that by a
\textit{code} (over $\Sigma$, with no further attributes) we
mean any finite non-empty subset of $\dfplus{\Sigma}\setminus\{\e\}$.

\begin{definition}[UD code]
Let $X$ be a code over $\Sigma$. $X$~is a \emph{uniquely
decipherable code}, if for any $x_1,\ldots,x_k$,
$y_1,\ldots,y_l\in X$ the equality $x_1\circ\cdots\circ
x_k=y_1\circ\cdots\circ y_l$ implies that $(x_1,\ldots,x_k)$ and
$(y_1,\ldots,y_l)$ are equal as sequences, \textit{i.e.}\ $k=l$
and $x_i=y_i$ for each $i\in\NN{k}$.
\end{definition}

\begin{definition}[UD $m$-code]
Let $X$ be a code over $\Sigma$. $X$~is a \emph{uniquely
decipherable $m$-code}, if for any $x_1,\ldots,x_k$,
$y_1,\ldots,y_l\in X$ the equality $x_1\circ_m\cdots\circ_m
x_k=y_1\circ_m\cdots\circ_m y_l$ implies that $(x_1,\ldots,x_k)$
and $(y_1,\ldots,y_l)$ are equal as sequences.
\end{definition}

In the remaining definitions, we use the obvious abbreviated notation.

\begin{definition}[MSD code and $m$-code]
Let $X$ be a code over $\Sigma$. $X$~is a \emph{multiset
decipherable code} (resp.\ \emph{$m$-code}), if for any
$x_1,\ldots,x_k$, $y_1,\ldots,y_l\in X$ the equality
$x_1\bullet\cdots\bullet x_k=y_1\bullet\cdots\bullet y_l$
implies that $\{\!\!\{x_1,\ldots,x_k\}\!\!\}$ and
$\{\!\!\{y_1,\ldots,y_l\}\!\!\}$ are equal as multisets.
\end{definition}

\begin{definition}[SD code and $m$-code]
Let $X$ be a code over $\Sigma$. $X$~is a \emph{set decipherable
code} (resp.\ \emph{$m$-code}), if for any $x_1,\ldots,x_k$,
$y_1,\ldots,y_l\in X$ the equality $x_1\bullet\cdots\bullet
x_k=y_1\bullet\cdots\bullet y_l$ implies that
$\{x_1,\ldots,x_k\}$ and $\{y_1,\ldots,y_l\}$ are equal as sets.
\end{definition}

\begin{definition}[ND code and $m$-code]
Let $X$ be a code over $\Sigma$. $X$~is a \emph{numerically
decipherable code} (resp.\ \emph{$m$-code}), if for any
$x_1,\ldots,x_k$, $y_1,\ldots,y_l\in X$ the equality
$x_1\bullet\cdots\bullet x_k=y_1\bullet\cdots\bullet y_l$
implies $k=l$.
\end{definition}

%
%
%
%
%

\begin{proposition}\label{prop:incl1}
If $X$ is a UD (resp.\ MSD, SD, ND) $m$-code, 
then $X$ is a UD (resp.\ MSD, SD, ND) code.
\end{proposition}

\begin{proof}
Assume $X$ is not a UD (resp.\ MSD, SD, ND) code. Then for some
$x_1,\ldots,x_k$, $y_1,\ldots,y_l\in X$ we have
$x_1\circ\cdots\circ x_k=y_1\circ\cdots\circ y_l$ with
$(x_1,\ldots,x_k)$ and $(y_1,\ldots,y_l)$ not satisfying the
final condition of the respective definition. But then,
irrespective of~$m$,
$x_1\circ_m\cdots\circ_m x_k=y_1\circ_m\cdots\circ_m y_l$
and $X$ is not a UD (resp.\ MSD, SD, ND) $m$-code.
\end{proof}

Note that the converse does not hold. A code may, for instance, fail to satisfy the 
UD $m$-code definition with 
$x_1\circ_m\cdots\circ_m x_k=y_1\circ_m\cdots\circ_m y_l$ 
and still be a UD code simply because some catenations in 
$x_1\circ\cdots\circ x_k$ and $y_1\circ\cdots\circ y_l$ are not defined.

\begin{example}
Take $X=\{x=
\begin{picture}(11,5)
    \dfimgbeginendlabel{00}{00}{a}
\end{picture}
\}$. $X$ is not a UD $m$-code, since $x\circ_m x=x$. It is a trivial UD code, though, because $x\circ x$ is not defined.
\end{example}

\begin{proposition}\label{prop:incl2}
Every UD code is an MSD code; every MSD code is an SD code and
an ND code. Every UD $m$-code is an MSD $m$-code; every MSD
$m$-code is an SD $m$-code and an ND $m$-code.
\end{proposition}

\begin{proof}
Obvious. 
\end{proof}

The diagram illustrates inclusions between different families of codes. 
A similar diagram can be made for $m$-codes.
Examples given below show that all those inclusions are strict.
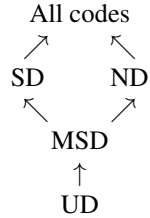
\begin{displaymath}
\begin{array}{c}
\mbox{All codes}\\
\nearrow~~~~~~~\nwarrow\\
\mbox{SD}~~~~~~~~~~\mbox{ND}\\
\nwarrow~~~~~~~\nearrow \\
\mbox{MSD}\\
\uparrow\\
\mbox{UD}
\end{array}
\end{displaymath}


\begin{example}
Four codes depicted below are, respectively, UD, MSD, SD and ND codes and $m$-codes. 
They are proper, in the sense that the MSD code is not a UD code 
(since $x_1\circ x_2\circ x_3\circ x_4 = x_2\circ x_4\circ x_1\circ x_3$), 
and the SD and ND codes are not MSD codes (since 
$y_1\circ y_4\circ y_4\circ y_3\circ y_2 = y_2\circ y_3\circ y_1\circ y_3\circ y_4\circ y_1$
and $z_1\circ z_3 = z_2\circ z_1$). 
For the sake of simplicity, we show sets that could also serve as examples
for corresponding properties of word codes. In fact, the MSD and SD examples come 
from~\cite{Guzman}.\\[2ex]
\begin{center}
\begin{tabular}{ll}
UD &
$w_1=\begin{picture}(20,5)
    \dfimgbeginlabel{00}{00}{a}
    \dfimgend{10}{00}
\end{picture}~~~
w_2=\begin{picture}(30,5)
    \dfimgbeginlabel{00}{00}{b}
    \dfimglabel{10}{00}{a}
    \dfimgend{20}{00}
\end{picture}~~~
w_3=\begin{picture}(30,5)
    \dfimgbeginlabel{00}{00}{b}
    \dfimglabel{10}{00}{b}
    \dfimgend{20}{00}
\end{picture}$\\[1ex]
MSD &
$x_1=\begin{picture}(40,5)
    \dfimgbeginlabel{00}{00}{a}
    \dfimglabel{10}{00}{a}
    \dfimglabel{20}{00}{b}
    \dfimgend{30}{00}
\end{picture}~~~
x_2=\begin{picture}(60,5)
    \dfimgbeginlabel{00}{00}{a}
    \dfimglabel{10}{00}{a}
    \dfimglabel{20}{00}{b}
    \dfimglabel{30}{00}{a}
    \dfimglabel{40}{00}{a}
    \dfimgend{50}{00}
\end{picture}~~~
x_3=\begin{picture}(40,5)
    \dfimgbeginlabel{00}{00}{a}
    \dfimglabel{10}{00}{b}
    \dfimglabel{20}{00}{a}
    \dfimgend{30}{00}
\end{picture}~~~
x_4=\begin{picture}(90,5)
    \dfimgbeginlabel{00}{00}{b}
    \dfimglabel{10}{00}{a}
    \dfimglabel{20}{00}{a}
    \dfimglabel{30}{00}{a}
    \dfimglabel{40}{00}{b}
    \dfimglabel{50}{00}{a}
    \dfimglabel{60}{00}{b}
    \dfimglabel{70}{00}{a}
    \dfimgend{80}{00}
\end{picture}$\\[1ex]
SD &
$y_1=\begin{picture}(20,5)
    \dfimgbeginlabel{00}{00}{a}
    \dfimgend{10}{00}
\end{picture}~~~
y_2=\begin{picture}(40,5)
    \dfimgbeginlabel{00}{00}{b}
    \dfimglabel{10}{00}{a}
    \dfimglabel{20}{00}{b}
    \dfimgend{30}{00}
\end{picture}~~~
y_3=\begin{picture}(60,5)
    \dfimgbeginlabel{00}{00}{a}
    \dfimglabel{10}{00}{a}
    \dfimglabel{20}{00}{b}
    \dfimglabel{30}{00}{a}
    \dfimglabel{40}{00}{a}
    \dfimgend{50}{00}
\end{picture}~~~
y_4=\begin{picture}(70,5)
    \dfimgbeginlabel{00}{00}{a}
    \dfimglabel{10}{00}{b}
    \dfimglabel{20}{00}{a}
    \dfimglabel{30}{00}{a}
    \dfimglabel{40}{00}{b}
    \dfimglabel{50}{00}{a}
    \dfimgend{60}{00}
\end{picture}$\\[1ex]
ND &
$z_1=\begin{picture}(20,5)
    \dfimgbeginlabel{00}{00}{a}
    \dfimgend{10}{00}
\end{picture}~~~
z_2=\begin{picture}(30,5)
    \dfimgbeginlabel{00}{00}{a}
    \dfimglabel{10}{00}{b}
    \dfimgend{20}{00}
\end{picture}~~~
z_3=\begin{picture}(30,5)
    \dfimgbeginlabel{00}{00}{b}
    \dfimglabel{10}{00}{a}
    \dfimgend{20}{00}
\end{picture}$
\end{tabular}
\end{center}
\end{example}
\medskip

Before proceeding with the main decidability results, note that
for UD, MSD and ND $m$-codes there is an ``easy case'' that can
be verified quickly just by analyzing the translation vectors of
figures. This is reflected in Theorem~\ref{thm:easyNonCodes}.

\begin{definition}[Two-sided and one-sided codes]
Let $X=\{x_1,\ldots,x_n\}$ be a code. If there exist non-negative integers
$\alpha_1,\ldots,\alpha_n$, not all equal to zero, such
that $\sum_{i=1}^n\alpha_i\dftransition{x_i}=(0,0)$,
then $X$ is called \emph{two-sided}. 
Otherwise, $X$~is called \emph{one-sided}.
\end{definition}

This condition can be interpreted geometrically as follows:
Translation vectors of a two-sided code do not fit in an open
half-plane. For a one-sided code, there exists a line passing
through $(0,0)$ such that all translation vectors are on one
side of it. Equivalently, there exists $\tau\in\integers^2$ such that 
the scalar products $\tau\cdot\dftransition{x_i}$ are all positive.

\begin{example}
The following set of figures is a two-sided code, with translation vectors 
$(1,2)$, $(1,-2)$ and $(-2,0)$:
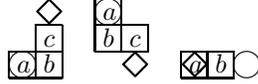
\begin{displaymath}
\begin{picture}(20,25)
    \dfimgbeginlabel{00}{00}{a}
    \dfimglabel{10}{00}{b}
    \dfimglabel{10}{10}{c}
    \dfimgend{10}{20}
\end{picture}~~~~~
\begin{picture}(20,25)
    \dfimgbeginlabel{00}{20}{a}
    \dfimglabel{00}{10}{b}
    \dfimglabel{10}{10}{c}
    \dfimgend{10}{00}
\end{picture}~~~~~
\begin{picture}(30,25)
    \dfimgendlabel{00}{00}{a}
    \dfimglabel{10}{00}{b}
    \dfimgbegin{20}{00}
\end{picture}
\end{displaymath}
It is a one-sided code, if the rightmost figure is removed.
\end{example}

\begin{theorem}[Necessary condition]\label{thm:easyNonCodes}
A two-sided code is not an ND $m$-code (and consequently neither 
an MSD nor UD $m$-code).
\end{theorem}

\begin{proof}
Assume $X=\{x_1,\ldots,x_n\}$ is two-sided, hence there exist
non-negative integers $\alpha_1,\ldots,\alpha_n$, not all equal to zero, 
such that $\sum_{i=1}^n\alpha_i\dftransition{x_i}=(0,0)$.
Let
\begin{displaymath}
    x=\underbrace{x_1\circ_m\cdots\circ_m x_1}_{\alpha_1\mathrm{~times}}\circ_m
      \underbrace{x_2\circ_m\cdots\circ_m x_2}_{\alpha_2\mathrm{~times}}\circ_m
      \cdots\circ_m
      \underbrace{x_n\circ_m\cdots\circ_m x_n}_{\alpha_n\mathrm{~times}}.
\end{displaymath}
Now consider the powers of $x$ (with respect to $\circ_m$),
$x^i$ for $i\ge 1$. Since $\dftransition{x}=(0,0)$, each of the powers
has the same domain. There is only a finite number of possible
labellings of this domain, which implies that regardless of the
merging function and labelling of~$x$, there exist
$p,q\in\naturals$, $p \neq q$ such that $x^p=x^q$. Hence $X$ is
not an ND $m$-code.
\end{proof}

\begin{corollary}
An ND $m$-code is one-sided.
\end{corollary}

\section{Decidability of verification}\label{sec:dec}

In this section we summarize all non-trivial decidability
results for the decipherability verification. We aim to prove
the decidability status for each combination of the
following orthogonal criteria: catenation type (with or without
a merging function), decipherability kind (UD, MSD, SD, ND) and
code geometry (one-sided, two-sided, two-sided with parallel
translation vectors). Two combinations remain open, however.

Proofs that have already appeared in our previous work and algorithms are omitted; 
references to respective papers are given. Note,
however, that in all decidable non-trivial cases there exist
algorithms to test the decipherability in question; the
algorithms effectively find a double factorization of a figure
if the answer is negative.

\subsection{Positive decidability results}

\begin{proposition}[see \cite{KolMoc}, Section 4]\label{prop:UDm}
Let $X$ be a one-sided code over~$\Sigma$. It is decidable whether $X$
is a UD $m$-code.
\end{proposition}

\begin{proposition}[see \cite{KolRAIRO}, Section 3]\label{prop:UD}
Let $X$ be a one-sided code over~$\Sigma$. It is decidable whether $X$
is a UD code.
\end{proposition}

Generalizing Propositions~\ref{prop:UDm} and~\ref{prop:UD}, we obtain a 
similar result for one-sided MSD, SD and ND codes and $m$-codes.

\begin{theorem}\label{th:oneSide}
Let $X$ be a one-sided code over~$\Sigma$. It is decidable whether $X$
is a \{UD, MSD, SD or ND\} \{code or $m$-code\}.
\end{theorem}

\begin{proof}
Starting with observations that allow us to
construct a ``bounding area'' for figures, we proceed with
properties that imply finiteness of possible configuration sets and,
consequently, decidability of the problem in question.

Let $X=\{x_1,\ldots,x_n\}\subseteq\dfplus{\Sigma}$ and let
$\dfbegin{x}=(0,0)$ for each $x\in X$. Since $X$ is one-sided,
there exists a vector $\tau$ such that for all $x\in X$,
\begin{displaymath}
\tau\cdot\dftransition{x}>0.
\end{displaymath}
We can assume that figures are sorted with respect to the angle of their translation
vectors in the following way:
\begin{displaymath}
\angle(\R{-\frac{\pi}{2}}{\tau},\dftransition{x_1}) ~~~\leq~~~
\angle(\R{-\frac{\pi}{2}}{\tau},\dftransition{x_2}) ~~~\leq~~~\ldots~~~\leq~~~
\angle(\R{-\frac{\pi}{2}}{\tau},\dftransition{x_n}).
\end{displaymath}

We choose constants $r_E,r_N,r_W,r_S > 0$ such that the vectors
\begin{eqnarray*}
\tau_E &=& r_E\tau,\\
\tau_N &=& r_N \R{\frac{\pi}{2}}{\dftransition{x_n}},\\
\tau_W &=& -r_W\tau,\\
\tau_S &=& r_S \R{-\frac{\pi}{2}}{\dftransition{x_1}}
\end{eqnarray*}
define a ``bounding area'' for figures in $X$, \textit{i.e.}, for all $x\in X$,
\begin{displaymath}
   \dfdomain{x}\cup\{\dfend{x}\}~~~\subseteq~~~
   \bigcap_{u \in\{\tau_E,\tau_N,\tau_W,\tau_S\}}\{\HP{u}{\dfbegin{x}}\}.
\end{displaymath}
The choice of $\tau$ determines a ``central axis'' along which figures will be
catenated. This is the line that bisects the half-plane containing all translation
vectors of figures in~$X$. Note that in all examples, $\tau$ and  $\tau_E$ are drawn 
as horizontal pointing eastwards, giving the natural meaning to the subscripts of  
$\tau_E$, $\tau_N$, $\tau_W$ and~$\tau_S$ vectors. The ordering of translation vectors
of figures in~$X$ is thus from the ``southernmost'' to ``northernmost''.

For $x\in\DFX{X}$ define
\begin{eqnarray*}
\CEP{x} &=& \HP{\tau_S}{\dfend{x}}\cap\HP{\tau_N}{\dfend{x}}\cap\HP{\tau_W}{\dfend{x}},\\
\CEN{x} &=& \integers^2\setminus\CEP{x},\\
\CWP{x} &=& \bigcup_{v}\{v+(\CEP{x}\cap\HP{\tau_E}{\dfend{x}})\},\\
\CWN{x} &=& \integers^2\setminus\CWP{x},
\end{eqnarray*}
where the union in the definition of $\CWP{x}$ is taken over $v\in\integers^2$ 
lying within an angle
spanned by vectors $-\dftransition{x_1}$ and~$-\dftransition{x_n}$.
Note that each term of the union is a trapezoid, resulting from the intersection
of four half-planes; see Figure~\ref{fig:hps} and Figure~\ref{fig:cwce}.

 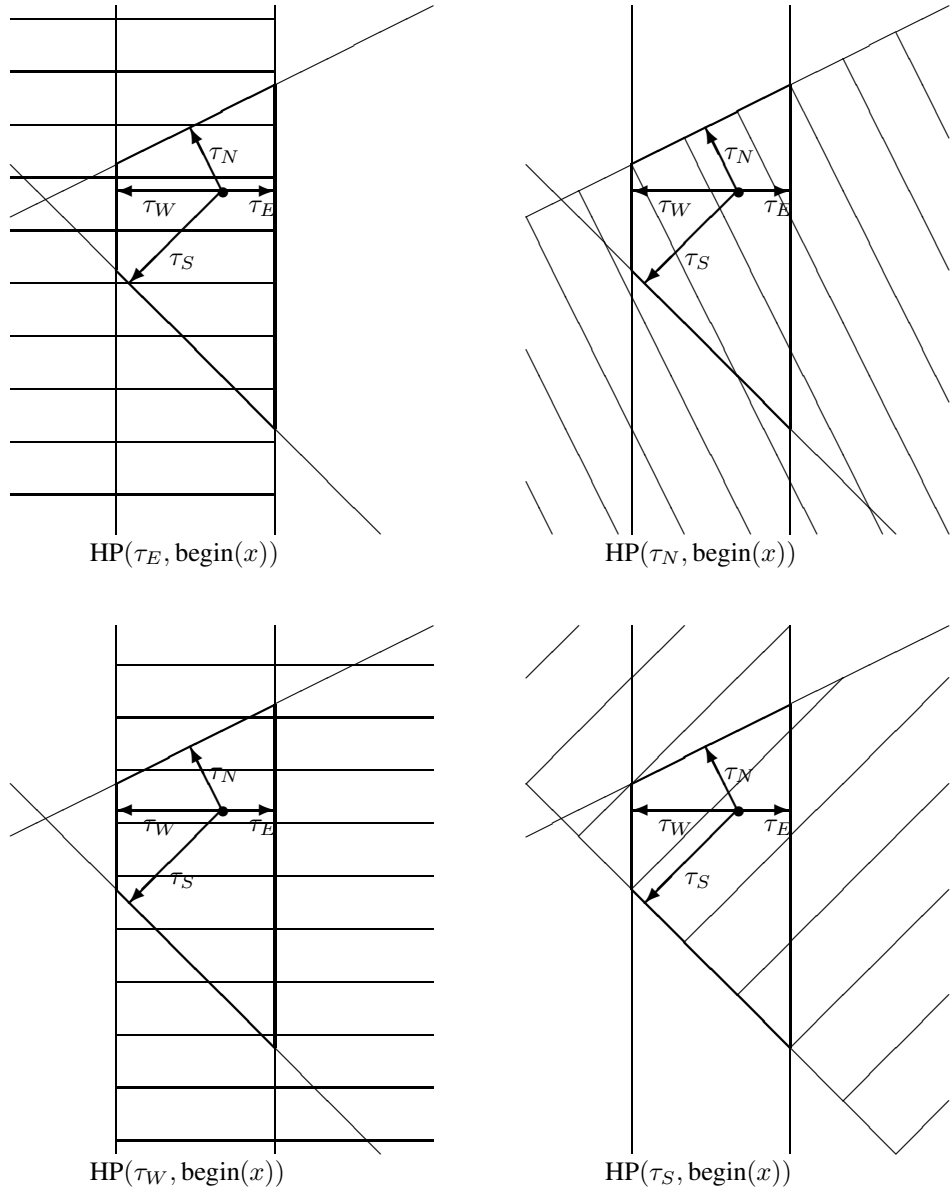
\begin{figure}[htp]
 \begin{center}
 \begin{picture}(100,220)
 \thinlines

 \put(-30,120){\line(2,1){160}}
 \put(10,200){\line(0,-1){200}}
 \put(70,200){\line(0,-1){200}}
 \put(-30,140){\line(1,-1){140}}

 \put(70,15){\line(-1,0){100}}
 \put(70,35){\line(-1,0){100}}
 \put(70,55){\line(-1,0){100}}
 \put(70,75){\line(-1,0){100}}
 \put(70,95){\line(-1,0){100}}
 \put(70,115){\line(-1,0){100}}
 \put(70,135){\line(-1,0){100}}
 \put(70,155){\line(-1,0){100}}
 \put(70,175){\line(-1,0){100}}
 \put(70,195){\line(-1,0){100}}

 \thicklines

 \put(10,100){\line(0,1){40}}
 \put(10,100){\line(1,-1){60}}
 \put(10,140){\line(2,1){60}}
 \put(70,40){\line(0,1){130}}

 \put(48,127){$\bullet$}

 \put(50,130){\vector(1,0){20}}\put(60,122){{$\tau_E$}}
 \put(50,130){\vector(-1,2){12}}\put(45,142){{$\tau_N$}}
 \put(50,130){\vector(-1,0){40}}\put(20,122){{$\tau_W$}}
 \put(50,130){\vector(-1,-1){35}}\put(30,102){{$\tau_S$}}

 \put(0,-10){$\HP{\tau_E}{\dfbegin{x}}$}

 \thinlines


 \end{picture}
 \hspace{20ex}
 \begin{picture}(100,220)
 \thinlines

 \put(-30,120){\line(2,1){160}}
 \put(10,200){\line(0,-1){200}}
 \put(70,200){\line(0,-1){200}}
 \put(-30,140){\line(1,-1){140}}

 \put(-30,20){\line(1,-2){10}}
 \put(-30,70){\line(1,-2){35}}
 \put(-30,120){\line(1,-2){60}}
 \put(-10,130){\line(1,-2){65}}
 \put(10,140){\line(1,-2){70}}
 \put(30,150){\line(1,-2){75}}
 \put(50,160){\line(1,-2){80}}
 \put(70,170){\line(1,-2){60}}
 \put(90,180){\line(1,-2){40}}
 \put(110,190){\line(1,-2){20}}

 \thicklines

 \put(10,100){\line(0,1){40}}
 \put(10,100){\line(1,-1){60}}
 \put(10,140){\line(2,1){60}}
 \put(70,40){\line(0,1){130}}

 \put(48,127){$\bullet$}

 \put(50,130){\vector(1,0){20}}\put(60,122){{$\tau_E$}}
 \put(50,130){\vector(-1,2){12}}\put(45,142){{$\tau_N$}}
 \put(50,130){\vector(-1,0){40}}\put(20,122){{$\tau_W$}}
 \put(50,130){\vector(-1,-1){35}}\put(30,102){{$\tau_S$}}

 \put(0,-10){$\HP{\tau_N}{\dfbegin{x}}$}

 \thinlines


 \end{picture}\\[3ex]
 \begin{picture}(100,220)
 \thinlines

 \put(-30,120){\line(2,1){160}}
 \put(10,200){\line(0,-1){200}}
 \put(70,200){\line(0,-1){200}}
 \put(-30,140){\line(1,-1){140}}

 \put(10,5){\line(1,0){120}}
 \put(10,25){\line(1,0){120}}
 \put(10,45){\line(1,0){120}}
 \put(10,65){\line(1,0){120}}
 \put(10,85){\line(1,0){120}}
 \put(10,105){\line(1,0){120}}
 \put(10,125){\line(1,0){120}}
 \put(10,145){\line(1,0){120}}
 \put(10,165){\line(1,0){120}}
 \put(10,185){\line(1,0){120}}

 \thicklines

 \put(10,100){\line(0,1){40}}
 \put(10,100){\line(1,-1){60}}
 \put(10,140){\line(2,1){60}}
 \put(70,40){\line(0,1){130}}

 \put(48,127){$\bullet$}

 \put(50,130){\vector(1,0){20}}\put(60,122){{$\tau_E$}}
 \put(50,130){\vector(-1,2){12}}\put(45,142){{$\tau_N$}}
 \put(50,130){\vector(-1,0){40}}\put(20,122){{$\tau_W$}}
 \put(50,130){\vector(-1,-1){35}}\put(30,102){{$\tau_S$}}

 \put(0,-10){$\HP{\tau_W}{\dfbegin{x}}$}

 \thinlines


 \end{picture}
 \hspace{20ex}
 \begin{picture}(100,220)
 \thinlines

 \put(-30,120){\line(2,1){160}}
 \put(10,200){\line(0,-1){200}}
 \put(70,200){\line(0,-1){200}}
 \put(-30,140){\line(1,-1){140}}

 \put(-30,180){\line(1,1){20}}
 \put(-30,140){\line(1,1){60}}
 \put(-10,120){\line(1,1){80}}
 \put(10,100){\line(1,1){80}}
 \put(30,80){\line(1,1){100}}
 \put(50,60){\line(1,1){80}}
 \put(70,40){\line(1,1){60}}
 \put(90,20){\line(1,1){40}}
 \put(110,00){\line(1,1){20}}

 \thicklines

 \put(10,100){\line(0,1){40}}
 \put(10,100){\line(1,-1){60}}
 \put(10,140){\line(2,1){60}}
 \put(70,40){\line(0,1){130}}

 \put(48,127){$\bullet$}

 \put(50,130){\vector(1,0){20}}\put(60,122){{$\tau_E$}}
 \put(50,130){\vector(-1,2){12}}\put(45,142){{$\tau_N$}}
 \put(50,130){\vector(-1,0){40}}\put(20,122){{$\tau_W$}}
 \put(50,130){\vector(-1,-1){35}}\put(30,102){{$\tau_S$}}

 \put(0,-10){$\HP{\tau_S}{\dfbegin{x}}$}

 \thinlines


 \end{picture}\\[3ex]
 \end{center}
 \caption{Half-planes $\HP{\tau}{\dfbegin{x}}$ for $\tau\in\{\tau_E,\tau_N,\tau_W,\tau_S\}$
 are marked with parallel lines; the black dot denotes the start point of~$x$.}
 \label{fig:hps}
 \end{figure}

 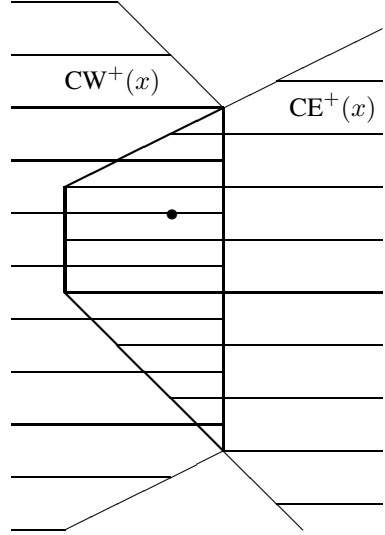
\begin{figure}[htp]
 \begin{center}
 \begin{picture}(100,220)

 \thicklines
 \put(10,100){\line(0,1){40}}
 \put(10,100){\line(1,-1){60}}
 \put(10,140){\line(2,1){60}}
 \put(70,40){\line(0,1){130}}

 \thinlines

 \put(48,127){$\bullet$}

 \put(70,40){\line(1,-1){30}}
 \put(70,40){\line(-2,-1){60}}
 \put(70,170){\line(-1,1){40}}
 \put(70,170){\line(2,1){60}}

 \thinlines
 \put(10,10){\line(-1,0){20}}
 \put(50,30){\line(-1,0){60}}
 \put(70,50){\line(-1,0){80}}
 \put(70,70){\line(-1,0){80}}
 \put(70,90){\line(-1,0){80}}
 \put(70,110){\line(-1,0){80}}
 \put(70,130){\line(-1,0){80}}
 \put(70,150){\line(-1,0){80}}
 \put(70,170){\line(-1,0){80}}
 \put(50,190){\line(-1,0){60}}
 \put(30,210){\line(-1,0){40}}

 \put(10,177){$\CWP{x}$}
 \put(90,20){\line(1,0){40}}
 \put(70,40){\line(1,0){60}}
 \put(50,60){\line(1,0){80}}
 \put(30,80){\line(1,0){100}}
 \put(10,100){\line(1,0){120}}
 \put(10,120){\line(1,0){120}}
 \put(10,140){\line(1,0){120}}
 \put(50,160){\line(1,0){80}}
 \put(90,180){\line(1,0){40}}

 \put(95,167){$\CEP{x}$}

 \end{picture}
 \end{center}
 \caption{$\CWP{x}$ and $\CEP{x}$ regions; the black dot denotes the end point of~$x$.}
\label{fig:cwce}
\end{figure}

Immediately from the definition we have following properties, for $x,y\in\DFX{X}$:
\begin{eqnarray*}
   \label{prop:ce-1}u\in \CEN{x}\cap\dfdomain{x}&\Rightarrow&\dflabel{x}(u)=\dflabel{x\bullet y}(u),\\
   \label{prop:ce-2}u\in \CEN{x}\setminus\dfdomain{x}&\Rightarrow&u\not\in\dfdomain{x\bullet y},\\
   u\in \CWN{x}&\Rightarrow&u\not\in\dfdomain{x},\\
   \CEP{x\bullet y}&\subseteq&\CEP{x},\\
   \CWP{x}&\subseteq&\CWP{x\bullet y}.
\end{eqnarray*}

For $x_1,\ldots,x_k$, $y_1,\ldots,y_l\in\DFX{X}$ we define a
\emph{configuration} as a pair of sequences
$((x_1,\ldots,x_k),(y_1,\ldots,y_l))$. A \emph{successor} of
such a configuration is either
$((x_1,\ldots,x_k,z),(y_1,\ldots,y_l))$ or
$((x_1,\ldots,x_k),(y_1,\ldots,y_l,z))$ for some $z\in X$. If a
configuration $C_2$ is a successor of $C_1$, we write $C_1\prec
C_2$. By $\prec^*$ we denote the transitive closure of $\prec$.

For a configuration $C=((x_1,\ldots,x_k),(y_1,\ldots,y_l))$ let
us denote:
\begin{eqnarray*}
	L(C)&=&\{x_1,\ldots,x_k\},\\
	L_{\bullet}(C)&=&x_1\bullet\ldots\bullet x_k,\\
	R(C)&=&\{y_1,\ldots,x_l\},\\
	R_{\bullet}(C)&=&y_1\bullet\ldots\bullet y_l.
\end{eqnarray*}

Now consider a starting configuration $((x),(y))$, for $x,y\in
X$, $x \neq y$. Assume that there exists a configuration $C$
such that $L_{\bullet}(C)=R_{\bullet}(C)$ and
$((x),(y))\prec^*C$. Now we have:
\begin{itemize}
\item $X$ is not a UD code (resp.\ UD $m$-code),
\item if $L(C)=R(C)$ as multisets then $X$ is not an MSD code (resp.\ MSD $m$-code),
\item if $L(C)=R(C)$ as sets then $X$ is not an SD code (resp.\ SD $m$-code),
\item if $|L(C)|=|R(C)|$ then $X$ is not an ND code (resp.\ ND $m$-code).
\end{itemize}
A configuration $C'$ such that $C'\prec^*C$ and $L_{\bullet}(C)=R_{\bullet}(C)$  for some $C$, is called a \emph{proper configuration}.

Our goal is either to show that there exists no proper
configuration, or to find such configuration(s). In the former
case, $X$ is a code (resp.\ $m$-code) of each kind. In the
latter case, if we find one of such configurations, $X$ is
already not a UD code (resp.\ UD $m$-code). To verify whether
$X$ is an MSD, SD or ND code (resp.\ $m$-code), we have to check
the above conditions for \emph{all} possible proper
configurations.

Let
\begin{displaymath}
\rho = \max_{x\in X}\min\{ n\in\naturals~|~\B{\dfbegin{x}}{n}\cap\dfdomain{x} \neq\emptyset\}.
\end{displaymath}
This number determines a distance within which both parts 
of a configuration, $L$ and $R$, can be found.
The following properties of a proper configuration $C$ are now easily verified:
\begin{eqnarray}
\label{prop:vertspan1}\B{\dfend{L_{\bullet}(C)}}{\rho} \cap (\CWP{R_{\bullet}(C)}\cup \CEP{R_{\bullet}(C)}) \neq \emptyset,\\
\label{prop:vertspan2}\B{\dfend{R_{\bullet}(C)}}{\rho} \cap (\CWP{L_{\bullet}(C)}\cup \CEP{L_{\bullet}(C)}) \neq \emptyset,
\end{eqnarray}
and for the common domain $D=\CEN{L_{\bullet}(C)}\cap\CEN{R_{\bullet}(C)}$:
\begin{equation}
\label{prop:red}\dflabel{L_{\bullet}(C)}\mid_{D}\equiv\dflabel{R_{\bullet}(C)}\mid_{D}.
\end{equation}

Notice that we do not need all of the information contained in
configurations, just those labellings that can be changed by
future catenations. By~(\ref{prop:red}), instead of a
configuration $C$ we can consider a \emph{reduced configuration}
defined as a pair
$(\pi_{RC}(L_{\bullet}(C),R_{\bullet}(C)),\pi_{RC}(R_{\bullet}(C),L_{\bullet}(C)))$
where
\begin{displaymath}
\pi_{RC}(z,z') = (\dfend{z},\dflabel{z}\mid_{\dfdomain{z}\setminus(\CEN{z}\cap \CEN{z'})}).
\end{displaymath}

Obviously we need only consider configurations where the span
along $\tau_E$ is bounded by $|\tau_E|$, \textit{i.e.},
\begin{displaymath}
|\tau_E\cdot(\dfend{L_{\bullet}(C)}-\dfend{R_{\bullet}(C)})|\leq |\tau_E|^2,
\end{displaymath}
since no single figure advances
$\dfend{L_{\bullet}(C)}$ or $\dfend{R_{\bullet}(C)}$ by more
than $|\tau_E|$. Moreover, (\ref{prop:vertspan1}) and
(\ref{prop:vertspan2}) restrict the perpendicular span (in the
direction of $\R{-\frac{\pi}{2}}{\tau_E}$). Hence the number of
reduced configurations, up to translation, is finite and there
is a finite number of proper configurations to check.
Consequently, we can verify whether $X$ is a UD, MSD, SD or ND
code (resp.\ $m$-code).
\end{proof}


Combined with Theorem~\ref{thm:easyNonCodes}, this proves the
decidability for all UD, MSD and ND $m$-codes. The case of two-sided SD
$m$-codes remains unsolved, however.

Two-sided codes with parallel translation vectors constitute an interesting special case.

\begin{definition}[Two-sided codes with parallel translation vectors]
Let $X=\{x_1,\ldots,x_n\}$ be a two-sided code. If there exists
a vector $\tau\in\integers^2$ and numbers $\alpha_1,\ldots,\alpha_n\in\integers$,
not all positive and not all negative, such that
$\dftransition{x_i}=\alpha_i\tau$ for $i=1,\ldots,n$,
then $X$ is called \emph{two-sided with parallel translation vectors}.
\end{definition}

\begin{proposition}[see \cite{KolCOCOON}, Section 4]\label{prop:UDpar}
Let $X$ be a two-sided code with parallel translation vectors. 
It is decidable whether $X$ is a UD code.
\end{proposition}

This can again be generalized to two-sided MSD, SD and ND codes
with parallel translation vectors:

\begin{theorem}\label{th:line}
Let $X$ be a two-sided code with parallel translation vectors. 
It is decidable whether $X$ is a UD, MSD, SD or ND code.
\end{theorem}

\begin{proof}
Even though the problem is one-dimentional, it cannot be easily transformed to any 
known word problem. Hence, a setting similar to that of Theorem~\ref{th:oneSide} is used:
we define bounding areas and use them to show that the number of possible
configurations is finite. This is accomplished by trying to find a figure that has two 
different factorizations and observing that the configurations are indeed bounded.

Let $X\subseteq\dfplus{\Sigma}$ be finite and non-empty and let
$\dfbegin{x}=(0,0)$ for each $x\in X$. Since translation vectors
of elements of $X$ are parallel, there exists a shortest vector
$\tau\in\integers^2$ such that for all $x\in X$,
\begin{displaymath}
	 \dftransition{x}\in \integers\tau=\{j\tau\mid j\in\integers\}.
\end{displaymath}
In particular, if $(t_1,t_2)=\dftransition{x}$ for some $x\in X$ with 
$\dftransition{x} \neq (0,0)$, then $\tau$ is one of the following vectors:
\begin{eqnarray}
\label{tau1}(t_1/\gcd(|t_1|,|t_2|)&,&t_2/\gcd(|t_1|,|t_2|)),\\
\label{tau2}(-t_1/\gcd(|t_1|,|t_2|)&,&-t_2/\gcd(|t_1|,|t_2|)),
\end{eqnarray}
where $\gcd$ denotes greatest common divisor. If all translation
vectors of elements of $X$ are $(0,0)$, then the decidability
problem is trivial: $X$ is an MSD, SD and ND code (since each
element can be used at most once) and $X$ is a UD code if and
only if no two elements can be concatenated, \textit{i.e.}\ no
two elements $x,y\in X$ have $\dfdomain{x}\cap\dfdomain{y}
\neq\emptyset$ (otherwise $xy=yx$); this case is obviously
decidable.

We define the following \emph{bounding areas}:
\begin{eqnarray*}
\lblock&=&\{u \in \integers^2 \mid 0 > u\cdot\tau \},\\
\block&=&\{u \in \integers^2 \mid 0 \leq u\cdot\tau < \tau\cdot\tau \},\\
\rblock&=&\{u \in \integers^2 \mid \tau\cdot\tau \leq u\cdot\tau \}.
\end{eqnarray*}

\begin{figure}
\begin{center}
\begin{picture}(120,60)
\thinlines

\put(80,0){\line(0,1){60}}
\put(12,0){\line(0,1){60}}

\put(10,27){$\bullet$}
\put(78,27){$\bullet$}
\put(12,30){\vector(1,0){67}}
\put(16,22){$(0,0)$}
\put(45,32){$\tau$}

\put(40,10){$\block$}
\put(-40,10){$\lblock$}
\put(100,10){$\rblock$}

\end{picture}
\end{center}
\caption{Bounding areas $\lblock$, $\block$  and $\rblock$.}
\end{figure}

For a non-empty figure $x\in\dfplus{\Sigma}$, \emph{bounding hulls} of $x$ are sets:
\begin{eqnarray*}
\hull{x} & = & \bigcup_{n=m\ldots M} \tr{n\tau}(\block),\\
\hullm{x} & = & \bigcup_{n=-M\ldots -m} \tr{n\tau}(\block),
\end{eqnarray*}
where
\begin{eqnarray*}
m&=&\min\{n\in \integers \mid \tr{n\tau}(\block) \cap (\dfdomain{x}\cup\{\dfbegin{x},\dfend{x}\}) \neq \emptyset \},\\
M&=&\max\{n\in \integers \mid \tr{n\tau}(\block) \cap (\dfdomain{x}\cup\{\dfbegin{x},\dfend{x}\}) \neq \emptyset \}.
\end{eqnarray*}
In addition, for the empty figure,
$\hull{\dfempty}  =  \emptyset$ and
$\hullm{\dfempty}  =  \emptyset$.

The area $\block$ is a vertical stripe of width equal to the length of~$\tau$. For a figure
$x\in\dfplus{\Sigma}$, $\hull{x}$ is a union of translated stripes such that the whole figure,
including its start and end points, lies inside it. The $\hullm{x}$ variant is a mirror image
of $\hull{x}$.

\textbf{Starting Configurations:}
Our goal is either to find a figure $x\in\dfplus{X}$ that has two different factorizations over 
elements of $X$, or to show that such a figure does not exist. If it exists, without loss 
of generality we can assume it has the following two different
\emph{$x$-} and \emph{$y$-factorizations}:
\begin{displaymath}
x=\dx_1\ddx_1\cdots\ddx_{k-1} \dx_k\ddx_k=\dy_1\ddy_1\cdots\ddy_{l-1} \dy_l\ddy_l
\end{displaymath}
where $\dx_1 \neq \dy_1$, $\dfbegin{\dx_1}=\dfbegin{\dy_1}=(0,0)$ and for $i\in\{1,\ldots,k\}$ and $j\in\{1,\ldots,l\}$
we have:
\begin{center}
\begin{tabular}{rclcrcl}
$\dx_i$  & $\in$ & $X$                          & \ and\  & $\hull{\dx_i}\cap\block$  & $\neq$ & $\emptyset$,\\
$\ddx_i$ & $\in$ & $\dfplus{X}\cup\{\dfempty\}$ & \ and\  & $\hull{\ddx_i}\cap\block$ & =      & $\emptyset$,\\
$\dy_j$  & $\in$ & $X$                          & \ and\  & $\hull{\dy_j}\cap\block$  & $\neq$ & $\emptyset$,\\
$\ddy_j$ & $\in$ & $\dfplus{X}\cup\{\dfempty\}$ & \ and\  & $\hull{\ddy_j}\cap\block$ & =      & $\emptyset$.
\end{tabular}
\end{center}
Observe that the following conditions for the $x$-factorization are satisfied for $i \in \{1,\ldots,k-1\}$:
\begin{itemize}
 \item if $\dfend{\dx_i} \in \lblock$,
 then $\dfbegin{\dx_{i+1}} \in \lblock$,
 \item if $\dfend{\dx_i} = (0,0)$,
 then $\ddx_i = \dfempty$ and $\dfbegin{\dx_{i+1}} = (0,0)$,
 \item if $\dfend{\dx_i} \in \rblock$,
 then $\dfbegin{\dx_{i+1}} \in \rblock$.
\end{itemize}
These are trivial implications of the assumption that $\hull{\ddx_i}\cap\block =
\emptyset$ and the fact that $\dx_i$ must be somehow linked with $\dx_{i+1}$.
Similar conditions are satisfied for the $y$-factorization.
In addition, the $x$-factorization must match the $y$-factorization, \textit{i.e.}:
\begin{itemize}
 \item if $\dfend{\dx_k} \in \lblock$,
 then $\dfend{\dy_l} \in \lblock$,
 \item if $\dfend{\dx_k} = (0,0)$,
 then $\dfend{\dy_l} = (0,0)$,
 \item if $\dfend{\dx_k} \in \rblock$,
 then $\dfend{\dy_l} \in \rblock$.
\end{itemize}
Also, it is clear that
\begin{eqnarray}
 \label{cond:b0dom}\bigcup_{i=1\ldots k}\dfdomain{\dx_i}\cap\block & = &
    \bigcup_{i=1\ldots l}\dfdomain{\dy_i}\cap\block,\\
 \label{cond:b0lab}\bigcup_{i=1\ldots k}\dflabel{\dx_i}\mid_{\block} & = &
    \bigcup_{i=1\ldots l}\dflabel{\dy_i}\mid_{\block}.
\end{eqnarray}

Now we consider all possible pairs of sequences $((\dx_i)_i,(\dy_j)_j)$ satisfying the above
conditions. Note that equality of such sequences is considered not up to translation: 
relative position of sequence elements is important. Such a pair will be called a
\emph{starting configuration}. Observe that there can be only a finite number of such
 configurations, since
\begin{eqnarray*}
 \bigcup_{i=1\ldots k}\dfdomain{\dx_i} & \subseteq & \bigcup_{x\in X} (\hull{x}\cup\hullm{x}),\\
 \bigcup_{i=1\ldots l}\dfdomain{\dy_i} & \subseteq & \bigcup_{x\in X} (\hull{x}\cup\hullm{x})
\end{eqnarray*}
and the set on the right hand side is bounded in the direction of $\tau$.
Also note that if there is no starting configuration for $X$, then obviously $X$ is a 
UD code and consequently an MSD, SD and ND code.

\textbf{Left and Right Configurations:}
We consider independently all starting configurations constructed for~$X$.
By (\ref{cond:b0dom}) and (\ref{cond:b0lab}), we can now forget the
labelling of $\block$. From a starting configuration $((\dx_i)_{i=1}^k,(\dy_j)_{j=1}^l)$ 
we construct \emph{L-} and
\emph{R-configurations} (left and right configurations)
\begin{eqnarray*}
	\conf{R} & = & ((\confDom{x}{R},\confLab{x}{R},\confBegends{x}{R}),(\confDom{y}{R},\confLab{y}{R},\confBegends{y}{R})),\\
	\conf{L} & = &
	((\confDom{x}{L},\confLab{x}{L},\confBegends{x}{L}),(\confDom{y}{L},\confLab{y}{L},\confBegends{y}{L})).
\end{eqnarray*}
First we show a construction for the $x$-part of a configuration:
\begin{center}
\begin{tabular}{rclcrcl}
 $\confDom{x}{R}$ & $=$ & $\bigcup_{i=1\ldots k}\dfdomain{\dx_i}\cap\rblock$ & and & $\confLab{x}{R}$ & $=$ & $\bigcup_{i=1\ldots
 k}\dflabel{\dx_i}\mid_{\rblock}$,\\
 $\confDom{x}{L}$ & $=$ & $\bigcup_{i=1\ldots k}\dfdomain{\dx_i}\cap\lblock$ & and & $\confLab{x}{L}$ & $=$ & $\bigcup_{i=1\ldots
 k}\dflabel{\dx_i}\mid_{\lblock}$
\end{tabular}
\end{center}
and multisets $\confBegends{x}{L}$, $\confBegends{x}{R}$ are obtained in the following way: for each
$i\in\{1,\ldots,k-1\}$:
\begin{itemize}
 \item if $\dfend{\dx_i} \in \lblock$,
 then $(\dfend{\dx_i},\dfbegin{\dx_{i+1}})$ is added to $\confBegends{x}{L}$,
 \item if $\dfend{\dx_i} = (0,0)$,
 then no pair is added to $\confBegends{x}{L}$ or $\confBegends{x}{R}$,
 \item if $\dfend{\dx_i} \in \rblock$,
 then $(\dfend{\dx_i},\dfbegin{\dx_{i+1}})$ is added to $\confBegends{x}{R}$
\end{itemize}
and
\begin{itemize}
 \item if $\dfend{\dx_k} \in \lblock$,
 then $(\dfend{\dx_k},\confEnd)$ is added to $\confBegends{x}{L}$,
 \item if $\dfend{\dx_k} = (0,0)$,
 then no pair is added to $\confBegends{x}{L}$ or $\confBegends{x}{R}$,
 \item if $\dfend{\dx_k} \in \rblock$,
 then $(\dfend{\dx_1},\confEnd)$ is added to $\confBegends{x}{R}$.
\end{itemize}
These multisets keep information on how figures $\dx_i$ and $\dx_{i+1}$ 
should be linked by $\ddx_i$ factors. 
The $\confEnd$ symbol denotes the end of the whole figure.

The $y$-part is created in a similar way.

\begin{example}
Consider a set containing the following figures (vertical lines separate the figures):
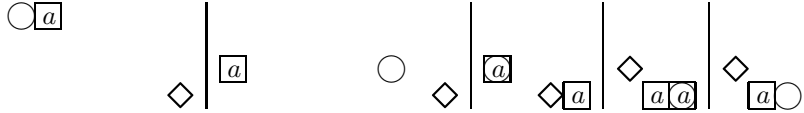
\begin{center}
\begin{picture}(300,40)

\dfimgbegin{00}{30}
\dfimglabel{10}{30}{a}
\dfimgend{60}{00}

\put(75,0){\line(0,1){40}}

\dfimglabel{80}{10}{a}
\dfimglabel{1300}{10}{a}
\dfimgbegin{140}{10}
\dfimgend{160}{00}

\put(175,0){\line(0,1){40}}

\dfimgbeginlabel{180}{10}{a}
\dfimgend{200}{00}
\dfimglabel{210}{00}{a}

\put(225,0){\line(0,1){40}}

\dfimgend{230}{10}
\dfimglabel{240}{00}{a}
\dfimgbeginlabel{250}{00}{a}

\put(265,0){\line(0,1){40}}

\dfimgend{270}{10}
\dfimglabel{280}{00}{a}
\dfimgbegin{290}{00}

\end{picture}
\end{center}
Taking $\tau=(2,-1)$, we construct one of possible starting
configurations ($x$-part only). We also show the construction of
the $x$-part of L- and R-configurations.

Figure~\ref{fig:conf} shows the construction. Each image
presents a current figure (with bold lines) and its translation
vector. Domain and labeling of all of the previous figures are
also presented, together with the end point of the previous
figure (which is important for the construction). $\block$ lies
between the slanted lines. Domains and labellings of L- and
R-configurations are presented in Fig.~\ref{fig:confdoms}.
\end{example}

\begin{figure}
\begin{center}
\begin{tabular}{lll}
$\dx_1$
&
\begin{picture}(150,50)
 \put(55,50){\line(-1,-2){25}}
 \put(80,50){\line(-1,-2){25}}
 \linethickness{1.5pt}
 \dfimgbeginlabel{50}{30}{a}
 \dfimglabel{80}{20}{a}
 \dfimgend{70}{20}
 \put(50,30){\vector(2,-1){20}}
\end{picture}
&
$\hull{\ddx_1}\subset\rblock$
\\
$\dx_2$
&
\begin{picture}(150,50)
 \put(55,50){\line(-1,-2){25}}
 \put(80,50){\line(-1,-2){25}}
 \dfimglabel{50}{30}{a}
 \dfimglabel{80}{20}{a}
 \dfimgend{70}{20}
 \linethickness{1.5pt}
 \dfimgbegin{90}{10}
 \dfimglabel{30}{10}{a}
 \dfimglabel{80}{10}{a}
 \dfimgend{110}{00}
 \put(90,10){\vector(2,-1){20}}
\end{picture}
&
\begin{tabular}{l}
$((2,-1),(4,-2))$ is added to $\confBegends{x}{R}$\\
$\hull{\ddx_2}\subset\rblock$
\end{tabular}
\\
$\dx_3$
&
\begin{picture}(150,50)
 \put(55,50){\line(-1,-2){25}}
 \put(80,50){\line(-1,-2){25}}
 \dfimglabel{50}{30}{a}
 \dfimglabel{80}{20}{a}
 \dfimglabel{80}{10}{a}
 \dfimgend{110}{00}
 \dfimglabel{30}{10}{a}
 \linethickness{1.5pt}
 \dfimgbeginlabel{70}{20}{a}
 \dfimglabel{60}{20}{a}
 \dfimgend{50}{30}
 \put(70,20){\vector(-2,1){20}}

\end{picture}
&
\begin{tabular}{l}
$((4,-2),(2,-1))$ is added to $\confBegends{x}{R}$\\
$\ddx_3=\dfempty$
\end{tabular}
\\
$\dx_4$
&
\begin{picture}(150,50)
 \put(55,50){\line(-1,-2){25}}
 \put(80,50){\line(-1,-2){25}}
 \dfimglabel{50}{30}{a}
 \dfimglabel{80}{20}{a}
 \dfimglabel{80}{10}{a}
 \dfimglabel{70}{20}{a}
 \dfimglabel{60}{20}{a}
 \dfimgend{50}{30}
 \dfimglabel{30}{10}{a}
 \linethickness{1.5pt}
 \dfimgbegin{50}{30}
 \dfimglabel{40}{30}{a}
 \dfimgend{30}{40}
 \put(50,30){\vector(-2,1){20}}
\end{picture}
&
\begin{tabular}{l}
no pair is added to $\confBegends{x}{L}$ and $\confBegends{x}{R}$\\
$\hull{\ddx_4}\subset\lblock$
\end{tabular}
\\
$\dx_5$
&
\begin{picture}(150,50)
 \put(55,50){\line(-1,-2){25}}
 \put(80,50){\line(-1,-2){25}}
 \dfimglabel{50}{30}{a}
 \dfimglabel{80}{20}{a}
 \dfimglabel{80}{10}{a}
 \dfimglabel{70}{20}{a}
 \dfimglabel{60}{20}{a}
 \dfimgbegin{50}{30}
 \dfimglabel{40}{30}{a}
 \dfimgend{30}{40}
 \dfimglabel{30}{10}{a}
 \linethickness{1.5pt}
 \dfimgbegin{30}{40}
 \dfimglabel{40}{40}{a}
 \dfimgend{90}{10}
 \put(30,40){\vector(2,-1){60}}
\end{picture}
&
\begin{tabular}{l}
\begin{tabular}{l}
$((-2,1),(-2,1))$ is added to $\confBegends{x}{L}$\\
$((4,-2),(\confEnd))$ is added to $\confBegends{x}{R}$\\
$\hull{\ddx_5}\subset\rblock$
\end{tabular}
\end{tabular}
\end{tabular}
 \caption{Construction of a sample starting configuration and its L- and R-configurations
 (figures added at each step are marked with thick lines).}
 \label{fig:conf}
\end{center}
\end{figure}
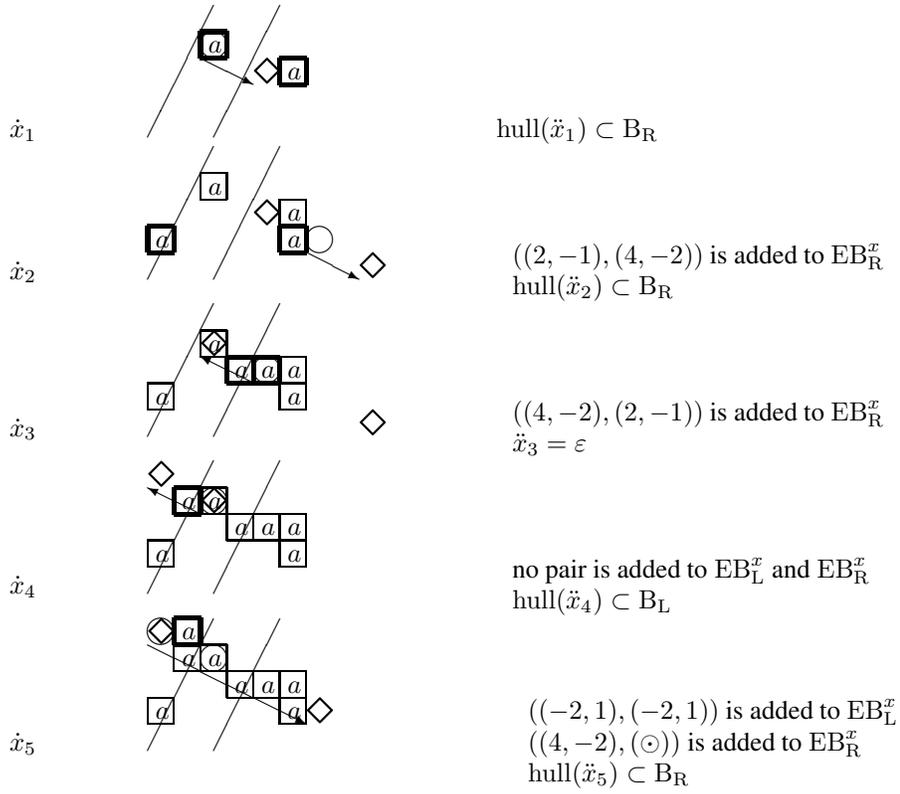

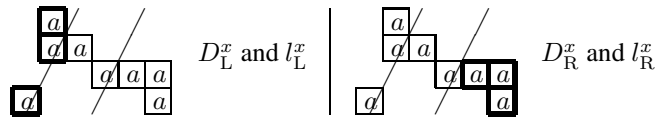
\begin{figure}
\begin{center}

\begin{picture}(230,40)
 \put(25,40){\line(-1,-2){20}}
 \put(50,40){\line(-1,-2){20}}
 \dfimglabel{20}{20}{a}
 \dfimglabel{50}{10}{a}
 \dfimglabel{50}{00}{a}
 \dfimglabel{40}{10}{a}
 \dfimglabel{30}{10}{a}
 \linethickness{1.5pt}
 \dfimglabel{00}{00}{a}
 \dfimglabel{10}{20}{a}
 \dfimglabel{10}{30}{a}
 \put(70,20){$\confDom{x}{L}$ and $\confLab{x}{L}$}
 \thinlines
 \put(120,00){\line(0,1){40}}

 \put(155,40){\line(-1,-2){20}}
 \put(180,40){\line(-1,-2){20}}
 \dfimglabel{130}{00}{a}
 \dfimglabel{150}{20}{a}
 \dfimglabel{140}{20}{a}
 \dfimglabel{140}{30}{a}
 \dfimglabel{160}{10}{a}
 \linethickness{1.5pt}
 \dfimglabel{180}{10}{a}
 \dfimglabel{180}{00}{a}
 \dfimglabel{170}{10}{a}
 \put(200,20){$\confDom{x}{R}$ and $\confLab{x}{R}$}

\end{picture}

 \caption{Domains and labellings of sample L- and R-configurations (respective objects 
 are marked with thick lines).}
 \label{fig:confdoms}
\end{center}
\end{figure}

Now let us consider the R-configuration only (the L-configuration is handled in a similar way). 
We say that an R-configuration
$((\confDom{x}{R},\confLab{x}{R},\confBegends{x}{R}),(\confDom{y}{R},\confLab{y}{R},
\confBegends{y}{R}))$ is \emph{terminating} if it
satisfies the following conditions:
\begin{itemize}
\item the domain and labelling of the $x$-part of the R-configuration match the domain and 
labelling of its $y$-part, \textit{i.e.},
\begin{displaymath}
 \confDom{x}{R} = \confDom{y}{R} \mbox{ and } \confLab{x}{R} = \confLab{y}{R},
\end{displaymath}
\item if a location of the end point of the whole figure is encoded in the R-configuration, then 
its location is the same in both $x$- and $y$-parts, \textit{i.e.}, for all $e\in\integers$,
\begin{displaymath}
 \label{term:end} (e,\confEnd)\in\confBegends{x}{R}\Leftrightarrow(e,\confEnd)\in\confBegends{y}{R},
\end{displaymath}
\item all points that should be linked together are trivially linked, since they are the same 
points, \textit{i.e.}, for all $(e,b)\in \confBegends{x}{R}\cup\confBegends{y}{R}$,
\begin{displaymath}
 \label{term:links2}e=b\mbox{ or }b=\confEnd.
\end{displaymath}
\end{itemize}

Note that if for some starting configuration we obtain a pair of terminating L- and R-
configurations, then $X$ is not a UD code (it can still be an MSD, SD or ND code, though). 
On the other hand, if we show that for all starting configurations such pair of terminating L- 
and R-configurations cannot be reached, then $X$ is a UD code (and hence an MSD, SD 
and ND code).

Similarly as in Theorem~\ref{th:oneSide}, to verify whether $X$ is an MSD, SD or ND code, 
we have to check the following conditions for all possible pairs $C$ of terminating L- and R-
configurations:
\begin{itemize}
\item if $\pi_x(C) = \pi_y(C)$ as multisets then $X$ is not an MSD code,
\item if $\pi_x(C) = \pi_y(C)$ as sets then $X$ is not an SD code,
\item if $|\pi_x(C)| = |\pi_y(C)|$ then $X$ is not an ND code,
\end{itemize}
where $\pi_x(C)$ and $\pi_y(C)$ denote respective multisets of elements used in the 
construction of $C$. Note that computation of $\pi_x(C)$ and $\pi_y(C)$ requires the history 
of $C$ to be kept; this does not spoil the finiteness of the part
of~$C$ that has to be kept.

\textbf{Obtaining New R-Configurations:}
When an R-configuration derived from a starting configuration is terminating, we can proceed 
to the analysis of the L-configuration. If the R-configuration is not terminating, we must check 
whether adding new figures may create a terminating configuration.

Initially such a derived configuration lies in $\rblock$. For simplicity of notation, we can 
translate such a configuration by a vector $-\tau$ (translating all its elements).

Now from the given R-configuration we want to obtain a new
R-configuration by adding new figures from $X$. In order to
obtain a new R-configuration from a given R-configuration, we
create the new R-configuration as a copy of the old one. Then
zero or more of the following operations must be performed (note
that they need not be admissible for an arbitrary
R-configuration or we may not need such operations to be
performed):
\begin{itemize}
 \item an $x$-part operation:
   add any $x\in X$ for which
   \begin{eqnarray}
     \label{extend:con1}\hull{x}\cap\lblock & = & \emptyset,\\
     \label{extend:con2}\hull{x}\cap\block	 & \neq & \emptyset,\\
     \label{extend:con3}\dfdomain{x}\cap\confDom{x}{R} & = & \emptyset
   \end{eqnarray}
	  to the new configuration, adding its domain and labelling function to the domain and 
	  labelling function of the $R$-configuration,
	  and replacing any pair $(e,b)$ from $\confBegends{x}{R}$ in the old configuration with 
	  two pairs $(e,\dfbegin{x})$ and $(\dfend{x},b)$
	  in the new one,
 \item an $y$-part operation: similarly.
\end{itemize}
In each step of creating the new generation of an R-configuration, we add only figures that 
change the given R-configuration
within $\block$; hence (\ref{extend:con2}). We add such figures to an R-configuration only at 
that step. In consecutive steps adding such figures is forbidden; hence (\ref{extend:con1}). At 
the first step this is a consequence
of restrictions for $\ddx_i$ and $\ddy_i$. Condition (\ref{extend:con3}) is obvious. Of course it 
is possible that a given R-configuration is not extendable at all.

After these operations we want the $x$-part of the R-configuration obtained to match its 
$y$-part on $\block$, \textit{i.e.},
\begin{displaymath}
 \confDom{x}{R}\cap\block    = \confDom{y}{R}\cap\block \mbox{~and~}
 \confLab{x}{R}\mid_{\block} = \confLab{y}{R}\mid_{\block}.
\end{displaymath}
In addition, for the $x$-part (and similarly for the $y$-part):
\begin{itemize}
 \item if $((0,0),b)\in\confBegends{x}{R}$,
 then $b=(0,0)$ or $b=\confEnd$,
 \item if $(e,(0,0))\in\confBegends{x}{R}$,
 then $e=(0,0)$,
\end{itemize}
and for both parts
\begin{itemize}
 \item $((0,0),\confEnd)\in\confBegends{x}{R}$ if and only if $((0,0),\confEnd)\in\confBegends{y}{R}$.
\end{itemize}
These conditions are trivial consequences of
(\ref{extend:con1}), (\ref{extend:con2}) and (\ref{extend:con3})
on new figures added to R-configuration. Of course it is
possible that one cannot obtain any R-configuration form the old
one.

Here, since the $x$-part and $y$-part of each newly created
R-configuration are the same, we now do not have to remember the
labelling of $\block$. When we forget this information,
configurations created lie in $\rblock$, so we can translate
them by $-\tau$ as previously.

Now observe that all parts of an R-configuration are bounded:
domains are contained in the area restricted by the widest hull
of elements of~$X$; multisets $\confBegends{x}{R}$ and
$\confBegends{y}{R}$ cannot be infinite, since eventually all
points must be linked. There are only finitely many such
configurations. Either we find a terminating R-configuration, or
we consider all configurations that can be obtained from a given
starting configuration performing one or more steps
described.
\end{proof}

Note that codes with parallel translation vectors are similar to classical
word codes and two-sidedness does not make a significant difference in terms
of decidability. This can be contrasted with the Post Correspondence Problem (PCP),
which is also ``linear'' yet undecidable. The essential difference is that PCP
configurations are extended with pre-defined pairs of words and there is no
\textit{a~priori} bound on how much two parts of a configuration can differ.
Code configurations are extended with individual words or figures and 
the respective bound can be determined by inspecting the size of words/figures.

\subsection{Negative decidability results}

\begin{proposition}[see \cite{KolRAIRO}, Section 2]\label{prop:UDundecid}
Let $X$ be a two-sided code over~$\Sigma$. It is undecidable whether $X$
is a UD code.
\end{proposition}

This result can again be extended to other decipherability kinds:

\begin{theorem}\label{th:undec}
Let $X$ be a two-sided code over~$\Sigma$. It is undecidable whether $X$
is a UD, MSD, SD or ND code.
\end{theorem}

\begin{proof}
We prove Theorem~\ref{th:undec} for UD codes first. The same reasoning is applied 
to MSD and SD codes, whilst for ND codes we use an additional technique,
described at the end of this proof. The proof is a reduction from
PCP to the decipherability problem. Given a PCP instance, we construct a two-sided
code such that the PCP instance has a solution if and only if the code is not decipherable.
Detailed explanation why this is indeed the case is given in the form of separate
Lemmas~\ref{post-code} and~\ref{code-post}, for the ``only if'' and ``if'' part, respectively.
They are, however, part of the proof since they rely heavily on
notations introduced here and would be impossible to formulate clearly
outside this context.

First we define figures that will be used throughout the reduction.
Let $\Sigma=\{a\}$.
For positive integers $h,h_N,h_E,h_S,h_W$ such that $h_N,h_E,h_S,h_W \leq h$ 
and $b,e\in\{N,E,S,W\}$ (with the usual geographical meaning) we define 
a \emph{directed hooked square} $\DHSh{h}{h_N}{h_E}{h_S}{h_W}{b}{e}$
to be a directed figure $f\in\dfplus{\Sigma}$ with:
\begin{eqnarray*}
	\dfdomain{f}&=&(B\setminus
	(H_N^-\cup H_E^-\cup H_S^-\cup H_W^-))\cup
	(H_N^+\cup H_E^+\cup H_S^+\cup H_W^+),\\
	\dfbegin{f}&=&\left\{
		\begin{array}{ll}
   		(0,h+2)&\mbox{if }b=N,\\
   		(h+2,0)&\mbox{if }b=E,\\
   		(0,-h-2)&\mbox{if }b=S,\\
   		(-h-2,0)&\mbox{if }b=W,\\
   	\end{array}\right.\\
	\dfend{f}&=&\left\{
		\begin{array}{ll}
   		(0,h+3)&\mbox{if }e=N,\\
   		(h+3,0)&\mbox{if }e=E,\\
   		(0,-h-3)&\mbox{if }e=S,\\
   		(-h-3,0)&\mbox{if }e=W,\\
   	\end{array}\right.
\end{eqnarray*}
where
\begin{eqnarray*}
	B&=&\{(x,y)\mid x,y\in\{-h-2,\ldots,h+2\}\},\\
	H_N^-&=&\{(-1,y)\mid y\in\{h+2-h_N,\ldots,h+2\}\}\cup\{(0,h+2-h_N)\},\\
	H_E^-&=&\{(x,1)\mid x\in\{h+2-h_E,\ldots,h+2\}\}\cup\{(h+2-h_E,0)\},\\
	H_S^-&=&\{(1,y)\mid y\in\{-h-2,\ldots,-h-2+h_S\}\}\cup\{(0,-h-2+h_S)\},\\
	H_W^-&=&\{(x,-1)\mid x\in\{-h-2,\ldots,-h-2+h_W\}\}\cup\{(-h-2+h_W,0)\},\\
	H_N^+&=&\{(1,y)\mid y\in\{h+3,\ldots,h+3+h_N\}\}\cup\{(0,h+3+h_N)\},\\
	H_E^+&=&\{(x,-1)\mid x\in\{h+3,\ldots,h+3+h_E\}\}\cup\{(h+3+h_E,0)\},\\
	H_S^+&=&\{(-1,y)\mid y\in\{-h-3-h_S,\ldots,-h-3\}\}\cup\{(0,-h-3-h_S)\},\\
	H_W^+&=&\{(x,1)\mid x\in\{-h-3-h_W,\ldots,-h-3\}\}\cup\{(-h-3-h_W,0)\},\\
\end{eqnarray*}
\textit{i.e.}\ $f$ is a square with hooks on each side (see
\textit{e.g.}\ Figure~\ref{dhs}).

\begin{figure}[htp]
\begin{center}
 \begin{picture}(260,190)
 \dfimglabel{50}{40}{a}\dfimglabel{50}{50}{a}\dfimglabel{50}{60}{a}\dfimglabel{50}{70}{a}\dfimglabel{50}{80}{a}\dfimglabel{200}{90}{a}\dfimglabel{50}{100}{a}\dfimglabel{50}{110}{a}\dfimglabel{50}{120}{a}\dfimglabel{50}{130}{a}\dfimglabel{50}{140}{a}\dfimglabel{50}{150}{a}\dfimglabel{50}{160}{a}\dfimglabel{60}{40}{a}\dfimglabel{60}{50}{a}\dfimglabel{60}{60}{a}\dfimglabel{60}{70}{a}\dfimglabel{60}{80}{a}\dfimglabel{190}{90}{a}\dfimglabel{60}{100}{a}\dfimglabel{60}{110}{a}\dfimglabel{60}{120}{a}\dfimglabel{60}{130}{a}\dfimglabel{60}{140}{a}\dfimglabel{60}{150}{a}\dfimglabel{60}{160}{a}\dfimglabel{70}{40}{a}\dfimglabel{70}{50}{a}\dfimglabel{70}{60}{a}\dfimglabel{70}{70}{a}\dfimglabel{70}{80}{a}\dfimglabel{180}{90}{a}\dfimglabel{70}{100}{a}\dfimglabel{70}{110}{a}\dfimglabel{70}{120}{a}\dfimglabel{70}{130}{a}\dfimglabel{70}{140}{a}\dfimglabel{70}{150}{a}\dfimglabel{70}{160}{a}\dfimglabel{80}{40}{a}\dfimglabel{80}{50}{a}\dfimglabel{80}{60}{a}\dfimglabel{80}{70}{a}\dfimglabel{80}{80}{a}\dfimglabel{110}{180}{a}\dfimglabel{80}{100}{a}\dfimglabel{80}{110}{a}\dfimglabel{80}{120}{a}\dfimglabel{80}{130}{a}\dfimglabel{80}{140}{a}\dfimglabel{80}{150}{a}\dfimglabel{80}{160}{a}\dfimglabel{90}{40}{a}\dfimglabel{90}{50}{a}\dfimglabel{90}{60}{a}\dfimglabel{90}{70}{a}\dfimglabel{90}{80}{a}\dfimglabel{120}{180}{a}\dfimglabel{120}{170}{a}\dfimglabel{90}{110}{a}\dfimglabel{90}{120}{a}\dfimglabel{90}{130}{a}\dfimglabel{90}{140}{a}\dfimglabel{90}{150}{a}\dfimglabel{90}{160}{a}\dfimglabel{100}{40}{a}\dfimglabel{100}{50}{a}\dfimglabel{100}{60}{a}\dfimglabel{100}{70}{a}\dfimglabel{100}{80}{a}\dfimglabel{100}{90}{a}\dfimglabel{100}{100}{a}\dfimglabel{100}{110}{a}\dfimglabel{100}{120}{a}\dfimglabel{100}{130}{a}\dfimglabel{100}{140}{a}\dfimglabel{00}{100}{a}\dfimglabel{40}{110}{a}\dfimglabel{110}{40}{a}\dfimglabel{110}{50}{a}\dfimglabel{110}{60}{a}\dfimglabel{200}{100}{a}\dfimglabel{110}{80}{a}\dfimglabel{110}{90}{a}\dfimglabel{110}{100}{a}\dfimglabel{110}{110}{a}\dfimglabel{110}{120}{a}\dfimglabel{110}{130}{a}\dfimglabel{110}{140}{a}\dfimglabel{30}{110}{a}\dfimglabel{110}{160}{a}\dfimglabel{100}{30}{a}\dfimglabel{100}{20}{a}\dfimglabel{100}{10}{a}\dfimglabel{100}{00}{a}\dfimglabel{120}{80}{a}\dfimglabel{120}{90}{a}\dfimglabel{120}{100}{a}\dfimglabel{120}{110}{a}\dfimglabel{120}{120}{a}\dfimglabel{120}{130}{a}\dfimglabel{120}{140}{a}\dfimglabel{120}{150}{a}\dfimglabel{120}{160}{a}\dfimglabel{130}{40}{a}\dfimglabel{130}{50}{a}\dfimglabel{130}{60}{a}\dfimglabel{130}{70}{a}\dfimglabel{130}{80}{a}\dfimglabel{130}{90}{a}\dfimglabel{130}{100}{a}\dfimglabel{130}{110}{a}\dfimglabel{130}{120}{a}\dfimglabel{130}{130}{a}\dfimglabel{130}{140}{a}\dfimglabel{130}{150}{a}\dfimglabel{130}{160}{a}\dfimglabel{140}{40}{a}\dfimglabel{140}{50}{a}\dfimglabel{140}{60}{a}\dfimglabel{140}{70}{a}\dfimglabel{140}{80}{a}\dfimglabel{140}{90}{a}\dfimglabel{140}{100}{a}\dfimglabel{140}{110}{a}\dfimglabel{140}{120}{a}\dfimglabel{140}{130}{a}\dfimglabel{140}{140}{a}\dfimglabel{140}{150}{a}\dfimglabel{140}{160}{a}\dfimglabel{150}{40}{a}\dfimglabel{150}{50}{a}\dfimglabel{150}{60}{a}\dfimglabel{150}{70}{a}\dfimglabel{150}{80}{a}\dfimglabel{150}{90}{a}\dfimglabel{110}{00}{a}\dfimglabel{20}{110}{a}\dfimglabel{150}{120}{a}\dfimglabel{150}{130}{a}\dfimglabel{150}{140}{a}\dfimglabel{150}{150}{a}\dfimglabel{150}{160}{a}\dfimglabel{160}{40}{a}\dfimglabel{160}{50}{a}\dfimglabel{160}{60}{a}\dfimglabel{160}{70}{a}\dfimglabel{160}{80}{a}\dfimglabel{160}{90}{a}\dfimglabel{160}{100}{a}\dfimglabel{10}{110}{a}\dfimglabel{160}{120}{a}\dfimglabel{160}{130}{a}\dfimglabel{160}{140}{a}\dfimglabel{160}{150}{a}\dfimglabel{160}{160}{a}\dfimglabel{170}{40}{a}\dfimglabel{170}{50}{a}\dfimglabel{170}{60}{a}\dfimglabel{170}{70}{a}\dfimglabel{170}{80}{a}\dfimglabel{170}{90}{a}\dfimglabel{170}{100}{a}\dfimglabel{00}{110}{a}\dfimglabel{170}{120}{a}\dfimglabel{170}{130}{a}\dfimglabel{170}{140}{a}\dfimglabel{170}{150}{a}\dfimglabel{170}{160}{a}

 \dfimgbegin{110}{160}
 \dfimgend{180}{100}

 \put(230,80){\DHSP{1}{2}{3}{4}\DHSNE}

 \end{picture}
 \caption{$\DHSh{4}{1}{2}{3}{4}{N}{E}$; full and reduced graphical
representation.}
 \label{dhs}
\end{center}
\end{figure}
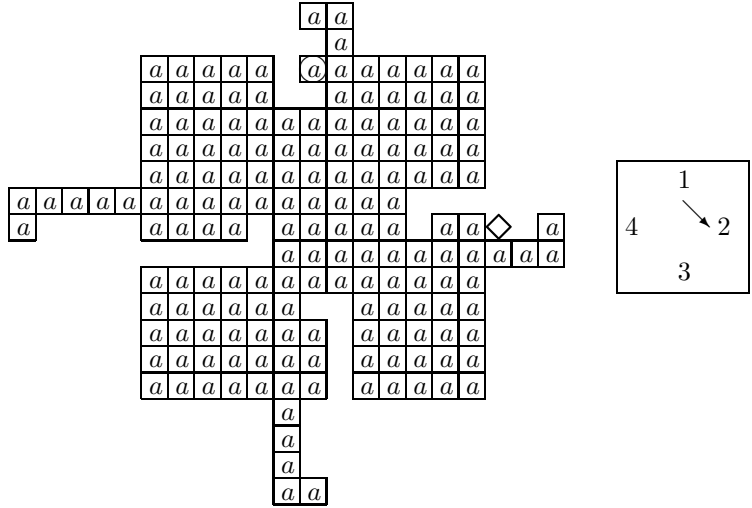

Observe that for
\begin{displaymath}
x=\DHSh{h}{h_N}{h_E}{h_S}{h_W}{b}{e} \mbox{~and~}
x'=\DHSh{h}{h'_N}{h'_E}{h'_S}{h'_W}{b'}{e'}
\end{displaymath}
catenation $x\circ x'$
is defined if and only if $e$ matches $b'$, \textit{i.e.},
\begin{center}
\begin{tabular}{rclcrcll}
$e$&=&$N$ & and & $b'$&=&$S$ & or\\
$e$&=&$E$ & and & $b'$&=&$W$ & or\\
$e$&=&$S$ & and & $b'$&=&$N$ & or\\
$e$&=&$W$ & and & $b'$&=&$E$ &
\end{tabular}
\end{center}
and $h_e=h'_{b'}$.

Now we encode a PCP instance in a set of
directed figures over $\Sigma=\{a\}$. The PCP can be stated as
follows: Let $A=\{a_1,\ldots,a_p\}$ be a finite alphabet,
$x_1,\ldots,x_k$, $y_1,\ldots,y_k\in A^+$ such that $x_i \neq
y_i$ for $i\in\NN{k}$. Find a sequence
$i_1,\ldots,i_n\in\NN{k}$, $n \geq 2$, such that $x_{i_1}\cdots
x_{i_n}=y_{i_1}\cdots y_{i_n}$.

We describe a set of directed figures $\mathbf{X}$ such that a given PCP
instance has a solution if and only if $\mathbf{X}$~is not a UD code.
Consider the following set:
\begin{displaymath}
H=\bigcup_{i\in\NN{k}}\{x_i,y_i,e_{x_i},e_{y_i},I_i\}\cup\{a_i\mid
i\in\NN{p}\}\cup\{x,y,x',y',b_x,b_y,e\},
\end{displaymath}
where $I_i$ are additional elements related to each pair
$(x_i,y_i)$ of the PCP instance. Set $h=|H|=5k+p+7$. We can
define a bijection between $H$ and $\NN{h}$, so from now on,
each element of $H$ is identified with its image by this
bijection. Since $h$ is now fixed, we write
$\DHS{h_N}{h_E}{h_S}{h_W}{b}{e}$ instead of
$\DHSh{h}{h_N}{h_E}{h_S}{h_W}{b}{e}$.

For each $x_i$, $i\in\NN{k}$, we define \emph{basic-figures} $[x_i[$, $]x_i[$ and
$]x_i]$ (Figure~\ref{basic-figures}); these figures will be
used to encode the word $x_i$ standing at the beginning (we call it
\emph{begin solution figure}), in the middle (\emph{middle
solution figure}) and at the end (\emph{end solution figure}) of the
PCP instance solution, respectively.
\begin{figure}[htp]
\begin{center}
\begin{tabular}{rc}
\raisebox{22pt}{$[x_i[^W_E=$}&
\begin{picture}(190,52)
\put(0,0){\DHSP{a_{i_1}}{x'}{b_x}{I_i}\DHSWE}
\put(53,22){$\circ$}
\put(60,0){\DHSP{a_{i_2}}{x'}{x'}{x'}\DHSWE}
\put(113,22){$\circ\ldots\circ$}
\put(140,0){\DHSP{a_{i_{r_i}}}{x'}{x'}{x'}\DHSWE}
\end{picture}
\\
\raisebox{22pt}{$]x_i[^W_E=$}&
\begin{picture}(190,52)
\put(0,0){\DHSP{a_{i_1}}{x'}{x_i}{x'}\DHSWE}
\put(53,22){$\circ$}
\put(60,0){\DHSP{a_{i_2}}{x'}{x'}{x'}\DHSWE}
\put(113,22){$\circ\ldots\circ$}
\put(140,0){\DHSP{a_{i_{r_i}}}{x'}{x'}{x'}\DHSWE}
\end{picture}
\\
\raisebox{22pt}{$]x_i]^W_S=$}&
\begin{picture}(190,52)
\put(0,0){\DHSP{a_{i_1}}{x'}{x'}{x'}\DHSWE}
\put(53,22){$\circ$}
\put(60,0){\DHSP{a_{i_2}}{x'}{x'}{x'}\DHSWE}
\put(113,22){$\circ\ldots\circ$}
\put(140,0){\DHSP{a_{i_{r_i}}}{e}{e_{x_i}}{x'}\DHSWS}
\end{picture}
\end{tabular}
 \caption{Basic-figures for $x_i=a_{i_1}\cdots a_{i_{r_i}}$.}
 \label{basic-figures}
\end{center}
\end{figure}
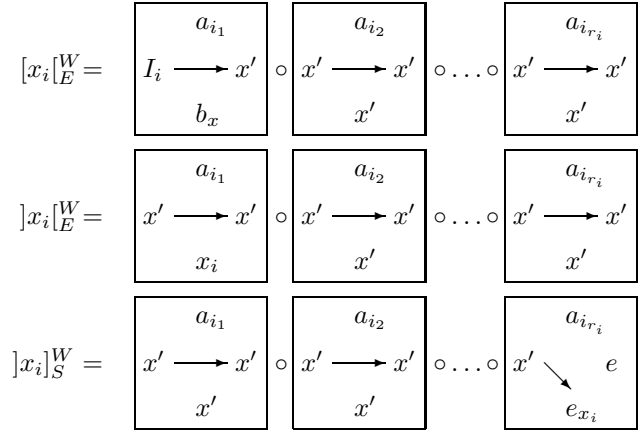

In addition we define \emph{annex-figures} (Figure~\ref{fig:annexFigures}).
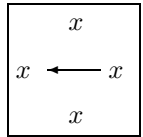
\begin{figure}[htp]
\begin{center}
\begin{tabular}{cccc}
\multicolumn{4}{l}{Annex-figures for passing information from north to south:}\\
\begin{picture}(50,52)
\DHSP{x_i}{x}{x_i}{x}\DHSWE
\end{picture}
&
\begin{picture}(50,52)
\DHSP{x_i}{x}{x_i}{x}\DHSEW
\end{picture}
&
\begin{picture}(50,52)
\DHSP{e_{x_i}}{e}{e_{x_i}}{x}\DHSNW
\end{picture}
&
\begin{picture}(50,52)
\DHSP{e_{x_i}}{e}{e_{x_i}}{x}\DHSWS
\end{picture}
\\
$M_x[i,N.S]^W_E$
&
$M_x[i,N.S]^E_W$
&
$E_x[i,N.S]^N_W$
&
$E_x[i,N.S]^W_S$
\\
&&&\\
\multicolumn{4}{l}{Annex-figures for passing information from north to west:}\\
\begin{picture}(50,52)
\DHSP{x_i}{x}{x}{x_i}\DHSWE
\end{picture}
&
\begin{picture}(50,52)
\DHSP{x_i}{x}{x}{x_i}\DHSEW
\end{picture}
&
\begin{picture}(50,52)
\DHSP{e_{x_i}}{e}{e}{e_{x_i}}\DHSNW
\end{picture}
&
\begin{picture}(50,52)
\DHSP{e_{x_i}}{e}{e}{e_{x_i}}\DHSWS
\end{picture}
\\
$M_x[i,N.W]^W_E$
&
$M_x[i,N.W]^E_W$
&
$E_x[i,N.W]^N_W$
&
$E_x[i,N.W]^W_S$
\\
&&&\\
\multicolumn{4}{l}{Annex-figures for passing information from east to west:}\\
\begin{picture}(50,52)
\DHSP{x}{x_i}{x}{x_i}\DHSWE
\end{picture}
&
\begin{picture}(50,52)
\DHSP{x}{x_i}{x}{x_i}\DHSEW
\end{picture}
&
\begin{picture}(50,52)
\DHSP{x}{e_{x_i}}{e}{e_{x_i}}\DHSWE
\end{picture}
&
\begin{picture}(50,52)
\DHSP{x}{e_{x_i}}{e}{e_{x_i}}\DHSEW
\end{picture}
\\
$M_x[i,E.W]^W_E$
&
$M_x[i,E.W]^E_W$
&
$E_x[i,E.W]^W_E$
&
$E_x[i,E.W]^E_W$
\\
\begin{picture}(50,52)
\DHSP{b_x}{x_i}{b_x}{I_i}\DHSES
\end{picture}
&
\begin{picture}(50,52)
\DHSP{b_x}{x_i}{b_x}{I_i}\DHSNE
\end{picture}
&
\begin{picture}(50,52)
\DHSP{b_x}{e_{x_i}}{e}{I_i}\DHSES
\end{picture}
&
\begin{picture}(50,52)
\DHSP{b_x}{e_{x_i}}{e}{I_i}\DHSNE
\end{picture}
\\
~$BM_x[i,E.W]^E_S$~
&
~$BM_x[i,E.W]^N_E$~
&
~$BE_x[i,E.W]^E_S$~
&
~$BE_x[i,E.W]^N_E$~
\\
&&&\\
\multicolumn{4}{l}{Annex-figures which pass no information:}\\
\begin{picture}(50,52)
\DHSP{x}{x}{x}{x}\DHSWE
\end{picture}
&
\begin{picture}(50,52)
\DHSP{x}{x}{x}{x}\DHSEW
\end{picture}
&&
\\
$N_x[]^W_E$
&
$N_x[]^E_W$
&&
\end{tabular}

 \caption{Annex-figures.}
 \label{fig:annexFigures}
\end{center}
\end{figure}

In the same way we define figures for the ``$y$-part'' of the PCP
instance, replacing the letter $x$ with~$y$.

Let $\mathbf{X}$ be the set of all figures defined ($6k$
basic-figures and $32k+2$ annex-figures, $16k$ for each part:
``$x$-part'' and ``$y$-part''). Observe that there exists no
half-plane of integer values anchored in $(0,0)$ (\textit{i.e.}\
$\{v\in\integers^2\mid u\cdot v>0\}$ for some $u\in\integers^2$) containing all
translation vectors of the figures we have defined.

The following two lemmas now complete the proof of Theorem~\ref{th:undec}
for UD, MSD and SD cases.

\begin{lemma}\label{post-code}
If the PCP instance has a solution then $\mathbf{X}$ is not a UD (MSD, SD) code.
\end{lemma}

\begin{proof}[of Lemma~\ref{post-code}]
Let $i_1,\ldots,i_n\in\NN{k}$ be a solution of the PCP instance, \textit{i.e.}\
$x_{i_1}\cdots x_{i_n}=y_{i_1}\cdots y_{i_n}$.
Consider the following directed figures:
\begin{eqnarray*}
 wx_1&=&[x_{i_1}[^W_E\circ]x_{i_2}[^W_E\circ\cdots\circ]x_{i_{n-1}}[^W_E\circ]x_{i_n}]^W_S,
\end{eqnarray*}
\begin{eqnarray*}
 wx_j&=&E_x[i_n,N.S]^N_W\circ\circTimes{N_x[]^E_W}{|x_n|-2}\circ\\
 	&&\circTimes{N_x[]^E_W}{|x_{n-1}|-1}\circ M_x[i_{n-1},N.E]^E_W\circ\\
 	&&\cdots\\
	&&\circTimes{N_x[]^E_W}{|x_{j+1}|-1}\circ M_x[i_{j+1},N.E]^E_W\circ\\
	&&\circTimes{N_x[]^E_W}{|x_j|-1}\circ M_x[i_j,N.E]^E_W\circ\\
	&&\circTimes{M_x[i_j,E.W]^E_W}{|x_{i_1}\cdots x_{i_{j-1}}|-1}\circ
BM_x[i_j,E.W]^E_S\\
	&&\mbox{(for even $j<n$)},
\end{eqnarray*}
\begin{eqnarray*}
 wx_j&=&BM_x[i_j,E.W]^N_E\circ\circTimes{M_x[i_j,E.W]^W_E}{|x_{i_1}\cdots
x_{i_{j-1}}|-1}\circ\\
 	&&M_x[i_j,N.E]^W_E\circ\circTimes{N_x[]^W_E}{|x_j|-1}\circ\\
 	&&M_x[i_{j+1},N.E]^W_E\circ\circTimes{N_x[]^W_E}{|x_{j+1}|-1}\circ\\
 	&&\cdots\\
 	&&M_x[i_{n-1},N.E]^W_E\circ\circTimes{N_x[]^W_E}{|x_{n-1}|-1}\circ\\
 	&&\circTimes{N_x[]^W_E}{|x_n|-2}\circ E_x[i_n,N.S]^W_S\\
 	&&\mbox{(for odd $j<n$)},
\end{eqnarray*}
\begin{eqnarray*}
 wx_n&=&E_x[i_n,N.W]^N_W\circ\circTimes{E_x[i_n,E.W]^E_W}{|x_{i_1}\cdots
x_{i_n}|-2}\circ BE_x[i_n,E.W]^E_S\\
 &&\mbox{(if $n$ is even)},\\
 wx_n&=&BE_x[i_n,E.W]^N_E\circ\circTimes{E_x[i_n,E.W]^W_E}{|x_{i_1}\cdots
x_{i_n}|-2}\circ E_x[i_n,N.W]^W_S\\
 &&\mbox{(if $n$ is odd)}.
\end{eqnarray*}
\begin{eqnarray*}
\end{eqnarray*}

In the same way we define figures $wy_1,\ldots,wy_n$.

It is easy to see that $wx_1\circ\cdots\circ wx_n=wy_1\circ\cdots\circ
wy_n\subseteq\dfplus{X}$. Hence $\mathbf{X}$ is
not a UD code.\\
\textit{(End of proof of Lemma~\ref{post-code}.)}
\end{proof}

\begin{example}
Consider
\begin{eqnarray*}
\Sigma&=&\{a,b\},\\
X&=&(x_1,x_2,x_3)=(a,ba,bab),\\
Y&=&(y_1,y_2,y_3)=(ab,aba,b).
\end{eqnarray*}
We have $x_1x_2x_3=y_1y_2y_3$. Figure $f$ with two different tilings
with elements of $\mathbf{X}$ is presented in
Figure~\ref{fig:xTiling} and Figure~\ref{fig:yTiling} (where thick
arrows show the flow of
information through annex-figures).
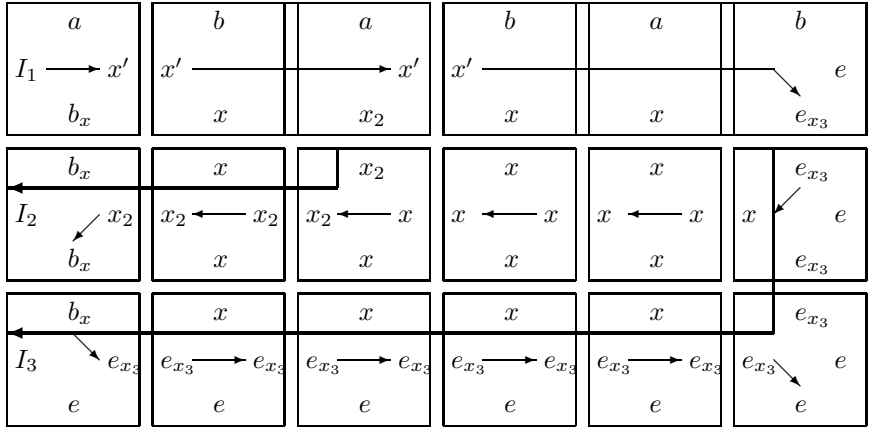
\begin{figure}[htp]
\begin{center}
\begin{picture}(325,160)
\put(0,0){\DHSP{b_x}{e_{x_3}}{e}{I_3}\DHSNE}
\put(55,0){\DHSP{x}{e_{x_3}}{e}{e_{x_3}}\DHSWE}
\put(110,0){\DHSP{x}{e_{x_3}}{e}{e_{x_3}}\DHSWE}
\put(165,0){\DHSP{x}{e_{x_3}}{e}{e_{x_3}}\DHSWE}
\put(220,0){\DHSP{x}{e_{x_3}}{e}{e_{x_3}}\DHSWE}
\put(275,0){\DHSP{e_{x_3}}{e}{e}{e_{x_3}}\DHSWS}

\put(0,55){\DHSP{b_x}{x_2}{b_x}{I_2}\DHSES}
\put(55,55){\DHSP{x}{x_2}{x}{x_2}\DHSEW}
\put(110,55){\DHSP{x_2}{x}{x}{x_2}\DHSEW}
\put(165,55){\DHSP{x}{x}{x}{x}\DHSEW}
\put(220,55){\DHSP{x}{x}{x}{x}\DHSEW}
\put(275,55){\DHSP{e_{x_3}}{e}{e_{x_3}}{x}\DHSNW}

\put(0,110){\DHSP{a}{x'}{b_x}{I_1}\DHSWE}

\put(55,110){\DHSP{b}{}{x}{x'}}

\put(110,110){\DHSP{a}{x'}{x_2}{}}
\put(105,110){\line(1,0){5}}
\put(105,160){\line(1,0){5}}
\put(70,135){\vector(1,0){75}}

\put(165,110){\DHSP{b}{}{x}{x'}}

\put(215,110){\line(1,0){5}}
\put(215,160){\line(1,0){5}}
\put(220,110){\DHSP{a}{}{x}{}}
\put(270,110){\line(1,0){5}}
\put(270,160){\line(1,0){5}}
\put(275,110){\DHSP{b}{e}{e_{x_3}}{}\DHSWS}
\put(180,135){\line(1,0){110}}

\thicklines
\put(125,90){\vector(-1,0){125}}
\put(125,90){\line(0,1){15}}

\put(290,35){\vector(-1,0){290}}
\put(290,35){\line(0,1){70}}

\end{picture}
 \caption{``$X$''-tiling of $f$.}
 \label{fig:xTiling}
\end{center}
\end{figure}

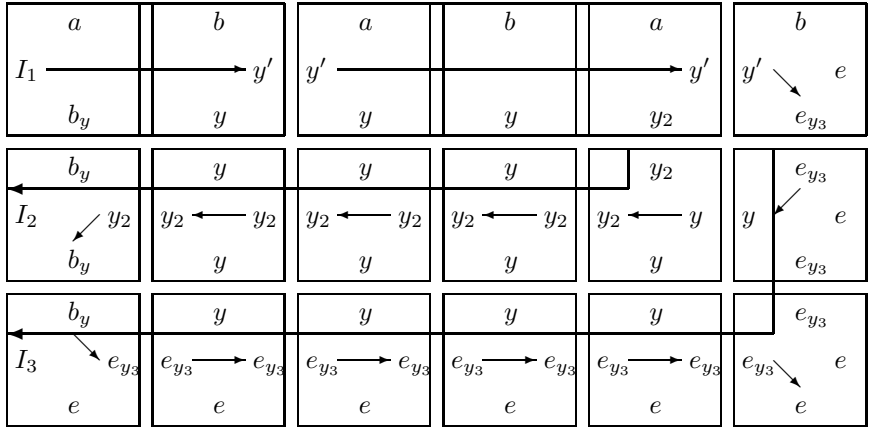
\begin{figure}[htp]
\begin{center}
\begin{picture}(325,160)
\put(0,0){\DHSP{b_y}{e_{y_3}}{e}{I_3}\DHSNE}
\put(55,0){\DHSP{y}{e_{y_3}}{e}{e_{y_3}}\DHSWE}
\put(110,0){\DHSP{y}{e_{y_3}}{e}{e_{y_3}}\DHSWE}
\put(165,0){\DHSP{y}{e_{y_3}}{e}{e_{y_3}}\DHSWE}
\put(220,0){\DHSP{y}{e_{y_3}}{e}{e_{y_3}}\DHSWE}
\put(275,0){\DHSP{e_{y_3}}{e}{e}{e_{y_3}}\DHSWS}

\put(0,55){\DHSP{b_y}{y_2}{b_y}{I_2}\DHSES}
\put(55,55){\DHSP{y}{y_2}{y}{y_2}\DHSEW}
\put(110,55){\DHSP{y}{y_2}{y}{y_2}\DHSEW}
\put(165,55){\DHSP{y}{y_2}{y}{y_2}\DHSEW}
\put(220,55){\DHSP{y_2}{y}{y}{y_2}\DHSEW}
\put(275,55){\DHSP{e_{y_3}}{e}{e_{y_3}}{y}\DHSNW}

\put(0,110){\DHSP{a}{}{b_y}{I_1}}
\put(50,110){\line(1,0){5}}
\put(50,160){\line(1,0){5}}
\put(55,110){\DHSP{b}{y'}{y}{}}
\put(15,135){\vector(1,0){75}}

\put(110,110){\DHSP{a}{}{y}{y'}}
\put(160,110){\line(1,0){5}}
\put(160,160){\line(1,0){5}}
\put(165,110){\DHSP{b}{}{y}{}}
\put(215,110){\line(1,0){5}}
\put(215,160){\line(1,0){5}}
\put(220,110){\DHSP{a}{y'}{y_2}{}}
\put(125,135){\vector(1,0){130}}

\put(275,110){\DHSP{b}{e}{e_{y_3}}{y'}\DHSWS}

\thicklines
\put(235,90){\vector(-1,0){235}}
\put(235,90){\line(0,1){15}}

\put(290,35){\vector(-1,0){290}}
\put(290,35){\line(0,1){70}}

\end{picture}
 \caption{``$Y$''-tiling of $f$.}
 \label{fig:yTiling}
\end{center}
\end{figure}
\end{example}

\begin{lemma}\label{code-post}
If $\mathbf{X}$ is not a UD (MSD, SD) code then the related PCP instance has a solution.
\end{lemma}

\begin{proof}[of Lemma~\ref{code-post}]
Let $f$ be a figure of minimal size (with respect to the size of its
domain) which admits two tilings with elements of $\mathbf{X}$,
\textit{i.e.}\ there exist $f_1,\ldots,f_p,g_1,\ldots,g_q\in
\mathbf{X}$ such that $f_1 \neq g_1$ and $f_1\circ\ldots\circ
f_p=g_1\circ\ldots\circ g_q$.

Consider directed hooked squares tiling the figure $f$ (these are
annex-figures and squares of which basic-figures are built). Let
$d$ be the westernmost among the northernmost of them. We have
following possibilities:
\begin{itemize}

\item\emph{Case 1:}
$d\in\bigcup_{z\in\{x,y\}}\bigcup_{j\in\NN{k}}\{
	E_z[j,N.S]^N_W,
	E_z[j,N.W]^N_W,
	E_z[j,E.W]^N_W
\}$

Since $d$ is the westernmost among the northernmost of all
squares tiling $f$, it cannot have north and west neighbour
squares, \textit{i.e.}\ squares hooked to it at the north and
west sides, respectively. Hence $f=d$, which contradicts the
definition of the double tiling of~$f$.

\item\emph{Case 2:}
$d\in\bigcup_{z\in\{x,y\}}\bigcup_{j\in\NN{k}}\{
	E_z[j,N.S]^W_S,
	E_z[j,N.W]^W_S
\}$

Since $d$ has no north and west neighbours, north and west hooks of
$d$ are uniquely determined by~$f$.
Each of figures listed is uniquely determined by its north and west hooks.
Hence $d$ is also uniquely determined by $f$. Now
$d$ has no west neighbour and it has the start point at its west side,
which implies that it must be the first one in a
sequence of figures whose catenation gives $f$, \textit{i.e.}\
$d=f_1=g_2$. Then either $f=d$ (contradiction as previously), or
$f'=f_2\circ\cdots\circ f_p=g_2\circ\cdots\circ g_q$ is a smaller
figure with two tilings, which contradicts
the minimality of~$f$.

\item\emph{Case 3:}
$d\in\bigcup_{z\in\{x,y\}}\bigcup_{j\in\NN{k}}\{
	M_z[j,N.S]^E_W,
	M_z[j,N.W]^E_W,
	M_z[j,E.W]^E_W
\}$\\ 
\mbox{}\hfill $\cup \{N_x[]^E_W,N_y[]^E_W\}$

As in \emph{Case 1}, $d$ is uniquely determined by $f$. Since
$d$ has no west neighbour and it has the end point at its west
side, it must be the last one in a sequence of figures whose
catenation gives $f$, \textit{i.e.}\ $d=f_p=g_q$. Then either
$f=d$ (contradiction as previously), or $f'=f_1\circ\cdots\circ
f_{p-1}=g_1\circ\cdots\circ g_{q-1}$ is a smaller figure with
two tilings, which contradicts the minimality of~$f$.

\item\emph{Case 4:}
$d\in\bigcup_{z\in\{x,y\}}\bigcup_{j\in\NN{k}}\{
	BM_z[j,E.W]^N_E,
	BE_z[j,E.W]^N_E
\}$

Now $d$ must be the first one in the tiling since it has the
start point at its north side and it is the northernmost in the
tiling. Observe that there exists no square with $e$-hook at the
north side. Hence $BM_z[j,E.W]^N_E$ and $BE_z[j,E.W]^N_E$
($z\in\{x,y\}$) cannot be the first elements of two different
tiling sequences of $f$. Consequently, $d$ is uniquely
determined by $f$ and $d=f_1=g_1$. Contradiction as in
\emph{Case~2}.

\item\emph{Case 5:}
$d\in\bigcup_{z\in\{x,y\}}\bigcup_{j\in\NN{k}}\{
	BM_z[j,E.W]^E_S,
	BE_z[j,E.W]^E_S
\}$

As in \emph{Case 4}, $d$ is uniquely determined by $f$. If
$d=BE_z[j,E.W]^E_S$ (for $z\in\{x,y\}$) then $d$ is the last
element of a tiling sequence. Contradiction as in \emph{Case 3}.
If $d=BM_z[j,E.W]^E_S$ (for $z\in\{x,y\}$) then (for some
$i\in\NN{p}$) $f=f_1\circ f_2\circ\cdots\circ f_{i-1}\circ
d\circ f_{i+1}\circ\cdots\circ f_p$, where
$f_1\in\{M_z[j,E.W]^E_W,M_z[j,N.W]^E_W\}$ and
$f_2=\cdots=f_{i-1}=M_z[j,E.W]^E_W$. Contradiction as in
\emph{Case~1}.

This leads us to a conclusion that:

\item\emph{Case 6:} Directed hooked square $d$ is a part of a
basic-figure. In particular, $d$ is a ``first part'' of $f_1$
and~$g_1$.

\end{itemize}

Now it is easy to observe the following properties of $f$'s tiling:
\begin{enumerate}

\item If $f_1$ is a figure that encodes one of the words from $X$,
then all $f_i$ ($i\in\NN{p}$) are figures encoding
``$x$-part'' of the related PCP instance (since there is no figure
that links a figure from ``$x$-part'' with a
figure from ``$y$-part''). In the same way, if $f_1$ encodes a word
from $Y$, then all $f_i$ encode ``$y$-part'' of the
PCP instance. A similar statement holds for $g_i$ ($i\in\NN{q}$).

\item First ``row'' of figures in the tiling is a sequence of middle
solution figures (may be empty) which is ended by
an end solution figure (that ends the row) and may be started with a
begin solution figure.

\item The sequence of middle solution figures from the first row
implies that in the tiling, leftmost column's hooks ($I_j$ hooks
of some $BM$ and $BE$ annex-figures) correspond to the sequence
of indices of words encoded by those figures.

\end{enumerate}
This leads us to a simple observation that the only possible two
tilings of $f$ are tilings of the form defined in
the proof of Lemma~\ref{post-code}. Hence the
related PCP instance has a solution.\\
\textit{(End of proof of Lemma~\ref{code-post}.)}
\end{proof}

\noindent\textit{(Proof of Theorem~\ref{th:undec}, continued.)}

Lemmas~\ref{post-code} and~\ref{code-post} complete the proof for 
UD, as well as MSD and SD codes, since it is clear
that exactly the same reasoning can be applied in the MSD and SD
cases. ND~codes, however, have to be dealt with separately, since both
factorizations have exactly the same number of figures. An additional technique
to handle the ND case is as follows: replace basic directed
hooked squares for both ``$x$-part'' and ``$y$-part'' with 25
squares. In the ``$x$-part'' the 25 squares will be connected
(into one figure), while in the ``$y$-part'' they will be
disconnected. See Figure~\ref{fig:replacementWE} and
Figure~\ref{fig:replacementNE}, where a construction is
presented for two kinds of figures. In both figures, $p$ and
$p_i$ are new symbols, different for each original directed
hooked square. Other kinds of figures can be dealt with in a
similar way.

Observe that the construction for UD, MSD and SD codes actually
uses vectors from a closed half-plane only. The construction for
ND codes can also be carried out in this way; however, more
complicated encoding figures are required then.

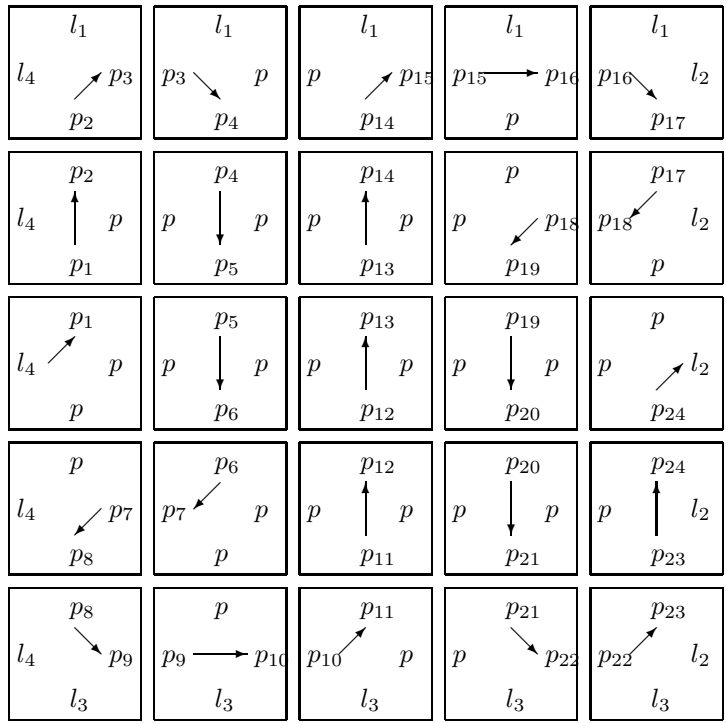
\begin{figure}[htp]
\begin{center}
\begin{picture}(270,270)

\put(0,220){\DHSP{l_1}{p_3}{p_2}{l_4}\DHSSE}
\put(55,220){\DHSP{l_1}{p}{p_4}{p_3}\DHSWS}
\put(110,220){\DHSP{l_1}{p_{15}}{p_{14}}{p}\DHSSE}
\put(165,220){\DHSP{l_1}{p_{16}}{p}{p_{15}}\DHSWE}
\put(220,220){\DHSP{l_1}{l_2}{p_{17}}{p_{16}}\DHSWS}

\put(0,165){\DHSP{p_2}{p}{p_1}{l_4}\DHSSN}
\put(55,165){\DHSP{p_4}{p}{p_5}{p}\DHSNS}
\put(110,165){\DHSP{p_{14}}{p}{p_{13}}{p}\DHSSN}
\put(165,165){\DHSP{p}{p_{18}}{p_{19}}{p}\DHSES}
\put(220,165){\DHSP{p_{17}}{l_2}{p}{p_{18}}\DHSNW}

\put(0,110){\DHSP{p_1}{p}{p}{l_4}\DHSWN}
\put(55,110){\DHSP{p_5}{p}{p_6}{p}\DHSNS}
\put(110,110){\DHSP{p_{13}}{p}{p_{12}}{p}\DHSSN}
\put(165,110){\DHSP{p_{19}}{p}{p_{20}}{p}\DHSNS}
\put(220,110){\DHSP{p}{l_2}{p_{24}}{p}\DHSSE}

\put(0,55){\DHSP{p}{p_7}{p_8}{l_4}\DHSES}
\put(55,55){\DHSP{p_6}{p}{p}{p_7}\DHSNW}
\put(110,55){\DHSP{p_{12}}{p}{p_{11}}{p}\DHSSN}
\put(165,55){\DHSP{p_{20}}{p}{p_{21}}{p}\DHSNS}
\put(220,55){\DHSP{p_{24}}{l_2}{p_{23}}{p}\DHSSN}

\put(0,0){\DHSP{p_8}{p_9}{l_3}{l_4}\DHSNE}
\put(55,0){\DHSP{p}{p_{10}}{l_3}{p_9}\DHSWE}
\put(110,0){\DHSP{p_{11}}{p}{l_3}{p_{10}}\DHSWN}
\put(165,0){\DHSP{p_{21}}{p_{22}}{l_3}{p}\DHSNE}
\put(220,0){\DHSP{p_{23}}{l_2}{l_3}{p_{22}}\DHSWN}

\end{picture}
\end{center}
\caption{Replacement figures for $\DHS{l_1}{l_2}{l_3}{l_4}{W}{E}$.}
\label{fig:replacementWE}
\end{figure}

\begin{figure}[htp]
\begin{center}
\begin{picture}(270,270)

\put(0,220){\DHSP{l_1}{p_8}{p_9}{l_4}\DHSES}
\put(55,220){\DHSP{l_1}{p}{p_7}{p_8}\DHSSW}
\put(110,220){\DHSP{l_1}{p}{p_{1}}{p}\DHSNS}
\put(165,220){\DHSP{l_1}{p_{22}}{p_{21}}{p}\DHSSE}
\put(220,220){\DHSP{l_1}{l_2}{p_{23}}{p_{22}}\DHSWS}

\put(0,165){\DHSP{p_9}{p}{p_{10}}{l_4}\DHSNS}
\put(55,165){\DHSP{p_7}{p}{p_6}{p}\DHSSN}
\put(110,165){\DHSP{p_{1}}{p}{p_{2}}{p}\DHSNS}
\put(165,165){\DHSP{p_{21}}{p}{p_{20}}{p}\DHSSN}
\put(220,165){\DHSP{p_{23}}{l_2}{p_{24}}{p}\DHSNS}

\put(0,110){\DHSP{p_{10}}{p}{p_{11}}{l_4}\DHSNS}
\put(55,110){\DHSP{p_6}{p}{p_5}{p}\DHSSN}
\put(110,110){\DHSP{p_{2}}{p}{p_{3}}{p}\DHSNS}
\put(165,110){\DHSP{p_{20}}{p}{p_{19}}{p}\DHSSN}
\put(220,110){\DHSP{p_{24}}{l_2}{p}{p}\DHSNE}

\put(0,55){\DHSP{p_{11}}{p}{p_{12}}{l_4}\DHSNS}
\put(55,55){\DHSP{p_5}{p_{4}}{p}{p}\DHSEN}
\put(110,55){\DHSP{p_{3}}{p}{p}{p_{4}}\DHSNW}
\put(165,55){\DHSP{p_{19}}{p_{18}}{p}{p}\DHSEN}
\put(220,55){\DHSP{p}{l_2}{p_{17}}{p_{18}}\DHSSW}

\put(0,0){\DHSP{p_{12}}{p_{13}}{l_3}{l_4}\DHSNE}
\put(55,0){\DHSP{p}{p_{14}}{l_3}{p_{13}}\DHSWE}
\put(110,0){\DHSP{p}{p_{15}}{l_3}{p_{14}}\DHSWE}
\put(165,0){\DHSP{p}{p_{16}}{l_3}{p_{15}}\DHSWE}
\put(220,0){\DHSP{p_{17}}{l_2}{l_3}{p_{16}}\DHSWN}

\end{picture}
\end{center}
\caption{Replacement figures for $\DHS{l_1}{l_2}{l_3}{l_4}{N}{E}$.}
\label{fig:replacementNE}
\end{figure}
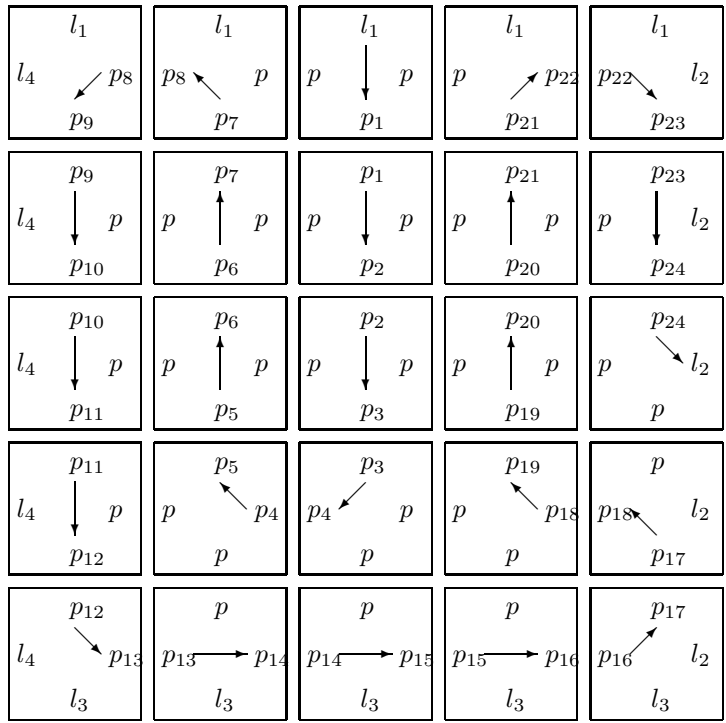

\end{proof}


%
%
%
%
%
%

\subsection{Summary of decidability results}
\label{subsec:summary}

The following table summarizes the status of decipherability decidability.
Decidable cases are marked with a~+, undecidable ones with a~$-$. Combinations
that are still open are denoted with a question mark.

\begin{center}
\begin{tabular}{|c|l|c|c|c|c|}
\hline
~~~~& & ~UD~ & ~MSD~ & ~ND~ & ~SD~\\
\hline
1 & ~One-sided codes & + & + & + & +\\
\hline
2 & ~One-sided $m$-codes & + & + & + & +\\
\hline
3 & ~Two-sided codes & $-$ & $-$ & $-$ & $-$\\
\hline
4 & ~Two-sided $m$-codes & + & + & + & ?\\
\hline
5 & ~Two-sided codes with parallel vectors & + & + & + & +\\
\hline
6 & ~Two-sided $m$-codes with parallel vectors~ & + & + & + & ?\\
\hline
\end{tabular}
\end{center}

\section{Final remarks}
\label{sec:final}

Note that the positive decidability cases depicted in lines 4
and~6 (of the table in Section~\ref{subsec:summary}) are trivial. By 
Theorem~\ref{thm:easyNonCodes}, two-sided
UD, MSD or ND $m$-codes do not exist. For other decidable
combinations, respective proofs lead to effective verification
algorithms.

On the other hand, the case of two-sided SD $m$-codes is
non-trivial; both SD and non-SD codes of this kind exist.
However, none of the proof techniques we have used so far can be
adapted to this case.

\bibliographystyle{abbrvnat}
\bibliography{wm-refs}

\end{document}